\documentclass[twoside,11pt]{article}
\usepackage{jmlr2e}

\usepackage{amsmath, amssymb} 
\usepackage{stmaryrd}
\usepackage{epsfig,calc, graphicx}
\usepackage{wasysym}
\usepackage[usenames]{color}
\usepackage{url}

\newcommand\ip[2]{\langle #1, #2\rangle}
\newcommand\natoms{K}
\newcommand\nsig{N}
\newcommand\sparsity{S}
\newcommand\ddim{d}

\newcommand\eps{\varepsilon}
\newcommand\epsmax{\varepsilon_{\max}}
\newcommand\epsmin{\varepsilon_{\min}}
\newcommand\dico{\Phi}
\newcommand\atom{\phi}
\newcommand\dicoset{\mathcal{D}}

\newcommand\pdico{\Psi}
\newcommand\ppdico{\bar\Psi}
\newcommand\patom{\psi}
\newcommand\ppatom{\bar\psi}
\newcommand\amp{c}

\newcommand\noise{r}

\newcommand\nsigma{\rho}
\newcommand\net{\mathcal{N}}

\newcommand\signop{\operatorname{sign}}
\newcommand\argmax{\operatorname{argmax}}

\newcommand{\R}{{\mathbb{R}}}

\newcommand{\E}{{\mathbb{E}}}

\renewcommand{\P}{{\mathbb{P}}}

\newenvironment{proofsketch}{\par\noindent{\bf Proof idea\ }}{\hfill\BlackBox\\[2mm]}

\jmlrheading{vol}{year}{pages}{1/14; Revised 7/14}{pub}{Karin Schnass}

\ShortHeadings{Local Identification of Overcomplete Dictionaries}{Karin Schnass}
\firstpageno{1}

\begin{document}

\title{Local Identification of Overcomplete Dictionaries}

\author{\name Karin Schnass \email karin.schnass@uibk.ac.at\\
\addr  Department of Mathematics\\
University of Innsbruck\\
 Technikerstra\ss e 19a\\
  6020 Innsbruck, Austria}

\editor{Shie Mannor}

\maketitle

\begin{abstract}%
This paper presents the first theoretical results showing that stable identification of overcomplete $\mu$-coherent dictionaries $\dico \in \R^{d\times K}$ is locally possible from training signals with sparsity levels $S$ up to the order $O(\mu^{-2})$ and signal to noise ratios up to $O(\sqrt{d})$. In particular the dictionary is recoverable as the local maximum of a new maximisation criterion that generalises the K-means criterion. For this maximisation criterion results for asymptotic exact recovery for sparsity levels up to $O(\mu^{-1})$ and stable recovery for sparsity levels up to $O(\mu^{-2})$ as well as signal to noise ratios up to $O(\sqrt{d})$ are provided. These asymptotic results translate to finite sample size recovery results with high probability as long as the sample size $N$ scales as $O(K^3dS \tilde \eps^{-2})$, where the recovery precision $\tilde \eps$ can go down to the asymptotically achievable precision. Further, to actually find the local maxima of the new criterion, a very simple Iterative Thresholding and K (signed) Means algorithm (ITKM), which has complexity $O(dKN)$ in each iteration, is presented and its local efficiency is demonstrated in several experiments.
\end{abstract}


\section{Introduction}\label{sec:intro}

Be it the 300 million photos uploaded to Facebook per day, the 800GB the large Hadron collider records per second or the 320.000GB per second it cannot record, it is clear that we have reached the age of big data. Indeed, in 2012, the amount of data existing worldwide is estimated to have reached 2.8 ZB = 2.800 billion GB and while 23 \% of these data are expected to be useful if analysed, only 1\% actually are. So how do we deal with this big data challenge?
The key concept, that has driven data processing and data analysis in the past decade, is that even high-dimensional data has intrinsically low complexity, meaning that every data point $y$ can be represented as linear combination of a sparse (small) number of elements or atoms $\atom_i \in \R^\ddim, \|\atom\|_2=1$ of an overcomplete dictionary $\dico = (\atom_1,\ldots \atom_\natoms)$, that is,
$$ y \approx \dico_I x_I= \sum_{i\in I} x(i) \varphi_i, $$ 
for a set $I$ of size $S$, $|I|=S$, which is small compared to
the ambient dimension, $S\ll d \leq K$.\\
These sparse components do not only describe the data but the representations can also be used for a myriad of efficient sparsity based data processing schemes, ranging from denoising, \citep{doelte06}, to  compressed sensing, \citep{do06cs,carota06}. Therefore a promising tool both for data analysis and data processing, that has emerged in the last years, is dictionary learning, also known as sparse coding or sparse component analysis. Dictionary learning addresses the fundamental question of how to automatically learn a dictionary, providing sparse representations for a given data class. That is, given $\nsig$ signals $y_n\in \R^\ddim$, stored as columns in a matrix $Y=(y_1,\ldots, y_\nsig)$, find a decomposition
$$ Y \approx \dico X$$
into a $\ddim \times \natoms$ dictionary matrix $\dico$ with unit norm columns and a $\natoms \times \nsig$ coefficient matrix with sparse columns. \\
Until recently the main research focus in dictionary learning has been on the development of algorithms.
Thus by now there is an ample choice of learning algorithms, that perform well in experiments and are popular in applications, \citep{olsfield96,  krra00, kreutz03, ahelbr06, yablda09, mabaposa10, sken10, rubrel10}. However, slowly the interest is shifting and researchers are starting to investigate also the theoretical aspects of dictionary learning. Following the first theoretical insights, originating in the blind source separation community, \citep{zipe01,gethci05}, there is now a set of generalisation bounds predicting how well a learned dictionary can be expected to sparsely approximate future data, \citep{mapo10, vamabr11, megr12, grjebaklse13}. These results give a theoretical foundation for dictionary learning as data processing tool, for example for compression, but unfortunately do not give guarantees that an efficient algorithm will find/recover a good dictionary provided that it exists. However, in order to justify the use of dictionary learning as data analysis tool, for instance in blind source separation, it is important to provide conditions under which an algorithm or scheme can identify the dictionary from a finite number of training signals, that is, the sources from the mixtures. Following the first dictionary identification results for the $\ell_1$-minimisation principle, which was suggested in \citet{zipe01}/\citet{pl07}, by \citet{grsc10, gewawrXX, bagrje14} and for the ER-SPuD algorithm for learning a basis in \citep{spwawr12}, 2013 has seen a number of interesting developments. First in \citet{sc14} it was shown that the K-SVD minimisation principle suggested in \citet{ahelbr06} can locally identify overcomplete tight dictionaries. Later in \citet{argemo13, aganne13} algorithms with {\it global} identification guarantees for coherent dictionaries were presented. Finally in \citet{aganjaneta13} it was shown that an alternating minimisation method is locally convergent to the correct generating dictionary. One aspect that all these results have in common is that the sparsity level of the training signals required for successful identification is of order $O(\mu^{-1})$ or $O(\sqrt{d})$ for incoherent dictionaries. Considering that on average sparse recovery in a given dictionary is successful for sparsity levels $O(\mu^{-2})$, \citep{tr08, scva07}, and that for dictionary learning we usually have a lot of training signals at our disposal, the same sparsity level should be sufficient for dictionary learning and indeed in this paper we provide the first indication that global dictionary identification could be possible for sparsity levels $O(\mu^{-2})$ by proving that it is locally possible. Further we show that in experiments a very simple iterative algorithm, based on thresholding and K signed means, is locally successful.\\
The paper is organised as follows. After introducing all necessary notation in Section~\ref{sec:notation} we present a new optimisation criterion, motivated by the analysis of the K-SVD principle, \citep{sc14}, in Section~\ref{sec:motivation}. In Section~\ref{sec:asymptotic} we give asymptotic identification results both for exact and stable recovery, which in Section~\ref{sec:finitesamples} are extended to results to finite sample sizes. Section~\ref{sec:experiments} provides an algorithm for actually finding a local optimum and some experiments confirming the theoretical results. Finally in the last section we compare the results for the new criterion to existing identification results, discuss the implications of these local results for global dictionary identification algorithms and point out directions for future research.

\section{Notations and Conventions}\label{sec:notation}
Before we jump into the fray, we collect some definitions and lose a few words on notations; usually subscripted letters will denote vectors with the exception of $\amp$ and $\eps$ where they are numbers, eg. $ (x_1,\ldots,x_\natoms)=X \in \R^{d\times \natoms}$ vs. $\amp=(\amp_1, \ldots, \amp_\natoms)\in \R^\natoms$, however, it should always be clear from the context what we are dealing with. \\
We consider a {\bf dictionary} $\dico$ a collection of $K$ unit norm vectors $\atom_i\in \R^d$, $\|\atom_i\|_2=1$. By abuse of notation we will also refer to the $d \times K$ matrix collecting the atoms as its columns as the dictionary, that is $\dico=(\atom_i, \ldots \atom_K)$. 
The maximal absolute inner product between two different atoms is called the {\bf coherence} $\mu$ of the dictionary, $\mu=\max_{i \neq j}|\ip{\atom_i}{\atom_j}|$. By $\dico_I$ we denote the restriction of the dictionary to the atoms indexed by $I$, that is $\dico_I=(\atom_{i_1}\ldots \atom_{i_\sparsity} )$, $i_j\in I$. We indicate the conjugate transpose of a matrix with a $^\star$, for example $\dico^\star$ would be the transpose of $\dico$.\\ 
The set of all dictionaries of a given size ($d \times K$) is denoted by $\mathcal{D}$.
For two dictionaries $\dico,\pdico\in \mathcal{D}$ we define the distance between each other as the maximal distance between two corresponding atoms, 
\begin{align}
d(\dico,\pdico):=\max_i \|\atom_i-\patom_i\|_2.\notag
\end{align}
We consider a {\bf frame} $F$ a collection of $K\geq d$ vectors $f_i\in\R^\ddim$ for which there exist two positive frame constants $A,B$ such that for all $v \in \R^\ddim$ we have
\begin{align}\label{framebound}
A \|v\|^2_2 \leq \sum_{i=1}^K |\ip{f_i}{v}|^2 \leq B \|v\|^2_2.
\end{align}
From \eqref{framebound} it follows that $F$, interpreted as $d\times K$ matrix, has rank $d$ and that its non-zero singular values are in the interval $[\sqrt{A}, \sqrt{B}]$. If $B$ can be chosen equal to $A$, that is $B=A$, the frame is called {\bf tight}.  If all frame elements $f_i$ have unit norm, we call $F$ a unit norm frame. For more details on frames, see e.g. \cite{ch03}.\\ 
Finally we introduce the Landau symbols $O,o$ to characterise the growth of a function. We write 
\begin{align}
 f(t)=O(g(t)) &\quad\mbox{ if }\quad \lim_{t \rightarrow 0/\infty} f(t)/g(t)= C<\infty \notag \\
\mbox{and}\quad f(t)=o(g(t))&\quad \mbox{ if }\quad \lim_{t \rightarrow 0/\infty} f(\eps)/g(\eps)=0.\notag
\end{align}.

\section{A Response Maximisation Criterion}\label{sec:motivation}
One of the origins of dictionary learning can be found in the field of vector quantisation, where the aim is to find a codebook (dictionary) such that the codewords (atoms) closely represent the data, 
that is
\begin{align}
\min_{\dico, X} \|Y-\dico X\|_F^2 \quad \mbox{s.t.} \quad x_n \in \{ e_i\}_i. \notag
\end{align}
Indeed the vector quantisation problem can be seen as an extreme case of dictionary learning, where we do not only want all our signals to be approximately $1$-sparse but also want the single non-zero coefficient equal to one. On the other hand we allow the atoms (codewords) to have any length. The problem above is usually solved by a K-means algorithm, which alternatively separates the training data into K clusters, each assigned to one codeword, and then updates the codeword to be the mean of the associated train signals. For more detailed information about vector quantisation or the K-means algorithm see for instance \cite{gega92} or the introduction of \cite{ahelbr06}.
If we relax the condition that each coefficient has to be positive, but in turn ask for the atoms to have unit norm, we are already getting closer to the concept of $1$-sparse dictionary learning,
\begin{align}
\min_{\dico\in \dicoset, X} \|Y-\dico X\|_F^2 \quad \mbox{s.t.} \quad x_n \in \{ \pm e_i\}_i. \notag
\end{align}
This minimisation problem can be rewritten as
\begin{align}
\min_{\dico \in \dicoset} \sum_n \min_{i,\sigma_i=\pm 1} \|y_n - \sigma_i \atom_i \|_2^2
& = \min_{\dico \in \dicoset} \sum_n \min_{i,\sigma_i =\pm 1} \|y_n\|^2_2 - 2 \sigma_i\ip{ y_n}{\atom_i} + \|\atom_i \|_2^2\notag \\
&= \|Y\|^2_F + N - 2 \max_{\dico \in \dicoset} \sum_n \max_i |\ip{ y_n}{\atom_i}|, \notag
\end{align}
and is therefore equivalent to the maximisation problem
\begin{align}\label{maxSis1}
\max_{\dico \in \dicoset} \sum_n \max_i  |\ip{ y_n}{\atom_i}|.
\end{align}
A local maximum of \eqref{maxSis1} can be found with a signed K-means algorithm, which assigns each training signal to the atom of the current dictionary giving the largest response in absolute value and then updating the atom as normalised signed mean of the associated training signals, see Section \ref{sec:experiments} for more details.
The question now is how do we go from these $1$-sparse dictionary learning formulations to $S$-sparse formulations with $S>1$. The most common generalisation, which provides the starting point for the MOD and the K-SVD algorithm, is to give up all constraints on the coefficients except for $S$-sparsity and to minimise,
\begin{align}\label{ksvdcrit}
(P_{P2})\hspace{2cm} \min_{\dico \in \dicoset, X} \|Y-\dico X\|_F^2 \quad \mbox{s.t.} \quad \|x_n\|_0 \leq S.
\end{align}
However, rewriting the problem we see that this formulation does not reduce to the same maximisation problem in case $S=1$. Then the best one term approximation in the given dictionary is simply the largest projection onto one atom and we have
\begin{align}
\min_{\dico \in \dicoset} \sum_n \min_{i, x_i} \|y_n - x_i \atom_i \|_2^2 &=
\min_{\dico \in \dicoset} \sum_n  \min_{i} \|y_n - \ip{\atom_i}{y_n} \atom_i \|_2^2\notag \\
& = \|Y\|_F^2 - \max_{\dico \in \dicoset} \sum_n\max_i  |\ip{ y_n}{\atom_i}|^2, \notag
\end{align}
leading instead to the maximisation problem 
\begin{align}
\max_{\dico \in \dicoset} \sum_n \max_i |\ip{ y_n}{\atom_i}|^2 \qquad \mbox{vs.} \qquad \max_{\dico \in \dicoset} \sum_n\max_i  |\ip{ y_n}{\atom_i}|. \notag
\end{align}
A local maximum can now be found using the same partitioning strategy as before but updating the atoms as largest singular vector rather than signed mean of the associated training signals, requiring K SVDs as opposed to K means. While the minimisation problem \eqref{ksvdcrit} is definitely the most effective generalisation for dictionary learning when the goal is compression, it brings with it some complications when used as analysis tool. Indeed in \cite{sc14} it has been shown that for $S=1$ the K-SVD criterion \eqref{ksvdcrit} can only identify the underlying dictionary from sparse random mixtures to arbitrary precision (given enough training samples) if this dictionary is tight and it is conjectured that the same holds for $S\geq 1$. Roughly simplified the reason for this is that for random sparse signals $\dico_I x_I$ and an $\eps$-perturbation $\pdico$ the average of the largest squared response behaves like 
\begin{align}
\frac{1}{2} \left(1-\frac{\eps^2}{2}+ c(\pdico)\right)^2 + \frac{1}{2} \left(1-\frac{\eps^2}{2} - c(\pdico)\right)^2= 1- \eps^2 + \frac{\eps^4}{4} + c(\pdico)^2. \notag
\end{align}
If $\dico$ is tight the term $c(\pdico)$ is constant over all dictionaries and therefore there is a local maximum at $\dico$.
From the above we also see that the average of the largest \emph{absolute} response should behave like
\begin{align}
\frac{1}{2} \left|1-\frac{\eps^2}{2}+ c(\pdico)\right| + \frac{1}{2} \left|1-\frac{\eps^2}{2} - c(\pdico)\right|= 1-\frac{\eps^2}{2}, \notag 
\end{align}
meaning that we should have a maximum at $\dico$ also if it is non-tight.
This suggests as alternative way to generalise the K-means optimisation principle for dictionary identification to simply maximise the absolute norm of the $S$-largest responses,
\begin{align}\label{thecriterion}
(P_{R1}) \hspace{2cm} \max_{\pdico \in \dicoset} \sum_n \max_{|I| = S} \| \pdico_I^\star y_n\|_1.
\end{align}
Other than for the K-SVD criterion it is not obvious that there should be a local optimum of \eqref{thecriterion} at $\dico$ even if all signals $y_n$ are perfectly $S$-sparse in $\dico$. Therefore it is quite intriguing that we will not only be able to prove local identifiability of any generating dictionary via \eqref{thecriterion} from randomly sparse signals, but that these identifiability properties are stable under coherence and noise. However, before we get to the main result in Theorem~\ref{th:stablefinite} on page \pageref{th:stablefinite}, we first have to lay the foundation, by providing suitable random signal models and by studying the asymptotic identifiability properties of the new principle.

\section{Asymptotic Results}   \label{sec:asymptotic}

We can get to the asymptotic version of the $S$-response maximisation principle in \eqref{thecriterion} simply by replacing the sum over the training signals with the expectation, leading to
\begin{align}\label{expcrit}
\max_{\pdico \in \mathcal D}\, \E_y \left( \max_{|I| = S} \| \pdico_I^\star y\|_1 \right).
\end{align}
Next we need a random sparse coefficient model to generate our signals $y$. 
We make the following definition, see also \cite{sc14}.
\begin{definition} 
A probability distribution (measure) $\nu$ on the unit sphere $S^{K-1} \subset \R^K$ is called symmetric if for all measurable sets $\mathcal{X}\subseteq S^{K-1}$, for all sign sequences $\sigma \in \{-1,1\}^K$ and all permutations $p$ we have
\begin{align}
\nu( \sigma \mathcal X)=\nu(\mathcal X), \quad &\mbox{where} \quad \sigma \mathcal X := \{ (\sigma_1 x_1, \ldots, \sigma_K x_K ) : x \in \mathcal{X} \},\quad \mbox{and}\notag\\
\nu( p( \mathcal X))=\nu(\mathcal X), \quad &\mbox{where} \quad p(\mathcal X ) := \{ ( x_{p(1)}, \ldots, x_{p(K)} ) : x \in \mathcal{X} \}\notag.
\end{align}
\end{definition}
Setting $y=\dico x$ where $x$ is drawn from a symmetric probability measure $\nu$ on the unit sphere has the advantage that for dictionaries which are orthonormal bases
the resulting signals have unit norm and for general dictionaries the signals have unit square norm in expectation, that is $\E (\|y\|_2^2)=1$. This reflects the situation in practical application, where we would normalise the signals in order to equally weight their importance.\\
One example of such a probability measure can be constructed from a non-negative, non-increasing sequence $\amp \in \R^\natoms$ with $\|c\|_2=1$, which we permute uniformly at random and provide with random $\pm$ signs. To be precise for a permutation $p:\{1,...,\natoms\}\rightarrow\{1,...,\natoms\}$ and a sign sequence $\sigma$, $\sigma_i=\pm 1$, we define the sequence $\amp_{p, \sigma}$ component-wise as
$\amp_{p, \sigma}(i):=\sigma_i \amp_{p(i)}$, and set 
$\nu(x)= (2^\natoms \natoms!)^{-1}$ if there exist $p, \sigma$ such that $x=\amp_{p, \sigma}$ and $\nu(x)=0$ otherwise.
While being very simple this measure exhibits all the necessary structure and indeed in our proofs we will reduce the general case of a symmetric measure to this simple case.\\
So far we have not incorporated any sparse structure in our coefficient distribution. 
To motivate the sparsity requirements on our coefficients we will recycle the simple negative example of a sparse coefficient distribution for which the original generating dictionary is not at a local maximum of \eqref{expcrit} with $S=1$ from \cite{sc14}.

\begin{example}\label{ex:flatcoeff}
Let $U$ be an orthonormal basis and let the signals be constructed as $y=\dico x$. If $x$ is randomly 2-sparse with 'flat' coefficients, that is, drawn from the simple symmetric probability measure with base sequence $c$, where $c_1=c_2=1/\sqrt{2}$, $c_i=0$ for $i\geq 3$, then $U$ is not a local maximum of~\eqref{expcrit} with $S=1$. \\
Indeed, since the signals are all 2-sparse, the maximal inner product with all atoms in $U$ is the same as the maximal inner product with only $d-1$ atoms. This degree of freedom we can use to construct an ascent direction. Choose $U_\eps=(u_1, \ldots, u_{d-1}, (u_d + \eps u_1)/\sqrt{1+\eps^2})$, then we have
\begin{align}
\E_y\left(\| U_\eps^\star y\|_\infty\right)&= \E_x\left(\left\| \left(x_1, \ldots, x_{d-1},  \frac{x_d + \eps x_1}{\sqrt{1+\eps^2}}\right)\right\|_\infty\right)\notag \\
&=\E_x \max \left\{\frac{1}{\sqrt{2}}, \left|  \frac{x_d + \eps x_1}{\sqrt{1+\eps^2}}\right|\right\}\notag \\
&=\frac{1}{\sqrt{2}}\left(1-\frac{1}{d(d-1)} + \frac{1}{d(d-1)}\frac{1+\eps}{\sqrt{1+\eps^2}}\right)\notag \\
&\geq \frac{1}{\sqrt{2}}\left(1+ \frac{1}{d(d-1)}\frac{\eps-\eps^2}{1+\eps^2}\right)>\frac{1}{\sqrt{2}} = \E_y\left(\| U^\star y\|_\infty \right).\notag
\end{align}
\end{example}
From the above example we see that, in order to have a local maximum of \eqref{expcrit} with $S=1$ at the original dictionary, we need our signals to be truly 1-sparse, that is, we need to have a decay between the first and the second largest coefficient. In the following sections we will study how large this decay should be to have a local maximum exactly at or near to the generating dictionary for more general dictionaries and sparsity levels.

\subsection{Exact Recovery}\label{sec:exact_rec}
To warm up we first provide an asymptotic exact dictionary identification result for~\eqref{expcrit} for incoherent dictionaries in the noiseless setting.
\begin{theorem}\label{th:exactasympt}
Let $\dico$ be a unit norm frame with frame constants $A\leq B$ and coherence $\mu$.
Let the coefficients $x$ be drawn from a symmetric probability distribution $\nu$ on the unit sphere $S^{\natoms-1} \subset \R^\natoms$ and assume that the signals are generated as $y=\dico x$.  If there exists $\beta>0$ such that
for $\amp_1(x)\geq \amp_2(x) \geq \ldots \geq \amp_K(x)\geq 0$ the non-increasing rearrangement of the absolute values of the components of $x$ we have
$ \amp_S(x)-\amp_{S+1}(x) - 2\mu \|x\|_1 \geq \beta $ almost surely, that is
\begin{align}\label{exactSdecaycond}
\nu \left(\amp_S(x)-\amp_{S+1}(x) - 2\mu \|x\|_1 \geq \beta\right) = 1,
\end{align} 
then there is a local maximum of \eqref{expcrit} at $\dico$. \\
Moreover for $\pdico \neq \dico$ we have $\E_y \left( \max_{|I| = S} \| \pdico_I^\star y\|_1 \right) < \E_y \left( \max_{|I| = S} \| \dico_I^\star y\|_1 \right) $ as soon as 
\begin{align}\label{epscondexactcont}
 \eps < \frac{\beta}{1 + 3\sqrt{ \log\left( \frac{25 K^2 S \sqrt{B}}{\beta(\bar c_1+ \ldots +\bar c_S)}\right) } },
 \end{align}
 where $\bar c_i :=\E_x (c_i(x))$.
\end{theorem}
\begin{proofsketch}We briefly sketch the main ideas of the proof, which are the same as for the corresponding theorem for the K-SVD principle in \cite{sc14}. For self-containedness of the paper the full proof is included in Appendix~\ref{app:exactsympt}. \\
Assume that we have the case of a simple probability measure based on one sequence c, that is $x=c_{p,\sigma}$.
For any fixed permutation $p$ the condition in~\eqref{exactSdecaycond} ensures that for all sign sequences $\sigma$, and consequently all signals, the maximal S responses of the original dictionary $\dico$ are attained at $I_p = p^{-1}\left(\{1\ldots S\}\right)$ and that there is a gap of size $\beta$ to the remaining responses.\\
For an $\eps$-perturbation of the generating dictionary we have $\patom_i \approx (1-\eps_i^2/2) \atom_i +\eps_i z_i$ for some unit vectors $z_i$ with $\ip{z_i}{\atom_i}=0$ and $\eps_i\leq\eps$. Now for most sign sequences the contribution of $\eps_i z_i$ to the response $\ip{\patom_i}{\dico c_{p,\sigma} }$ will be smaller than $\beta/2$ so the maximal S responses will still be attained at $I_p$. Comparing the loss of the perturbed dictionary over the typical sign sequences of all permutations, which scales as $\frac{(c_1 + \ldots+ c_S)}{2K} \sum \eps_i^2$, to the maximal gain $S\eps \sqrt{B}$ over the approximately $2\sum_i \exp\left(-\beta^2/\eps_i^2\right)$ atypical sign sequences shows that there is a maximum at the original dictionary. The general result follows from an integration argument.
\end{proofsketch}
As already mentioned, while for the K-SVD criterion~\eqref{ksvdcrit} there is always an optimum at the generating dictionary if all training signals are $S$-sparse, this it is not obvious for the response principle. Indeed, in the special case where all the training signals are exactly $S$-sparse, $c_{S+1}(x)=0$ almost surely, we get an additional condition to ensure asymptotic recoverability,
\begin{align}
\amp_S(x) - 2\mu \sum_{s=1}^{S} c_s(x) \geq \beta > 0, \qquad \mbox{almost surely.}\notag
\end{align}
To get a better feeling for this constraint we bound the sum over the S largest reponses by S times the largest response, $ \sum_{s=1}^{S} c_s(x) \leq S c_1(x)$ and arrive at the condition
\begin{align}\label{exactSsparsecond}
\frac{\amp_S(x)}{\amp_1(x)} \gtrsim 2 S \mu,
\end{align}
which is the classical condition under which simple thresholding will find the support of an exactly $S$-sparse signal, compare for instance \citet{scva08}. 


\subsection{Stability under Coherence and Noise}
While giving a first insight into the identification properties of the response principle, Theorem~\ref{th:exactasympt} suffers from two main limitations.\\
First, the required condition on the coherence of the dictionary with respect to the decay of the coefficients, $c_S(x)-c_{S+1}(x)- 2\mu\|c(x)\|_1>0$, is unfortunately quite strict. In the most favourable case of exactly $S$ sparse signals with equally sized coefficients, $c_1(x)=c_S(x)=1/\sqrt{S}$, we see from \eqref{exactSsparsecond} that we can only identify dictionaries from very sparse signals, where $S=\lesssim \mu^{-1}$. In case of very incoherent dictionaries with
$\mu = O(1/\sqrt{d})$ this means that $S \lesssim \sqrt{d}$. 
However, for most sign sequences $\sigma$ we have 
\begin{align}
|\langle \atom_{i}, \dico \amp_{p,\sigma} \rangle|&= \Big|  \sigma_{i}  \amp_{p(i)} + \sum_{j\neq i} \sigma_j \amp_{p(j)}\langle \atom_{i}, \atom_j  \rangle \Big|
 \approx \amp_{p(i)} \pm  \left( \sum_{j\neq i} \amp^2_{p(j)} |\ip{ \atom_{i}}{ \atom_j} |^2 \right)^{1/2} \approx  \amp_{p(i)} \pm \mu,\notag
\end{align}
which indicates that a condition of the form $\mu \lesssim c_S-c_{S+1}$ may be strong enough to guarantee (approximate) recoverability of the dictionary. Assuming again the most favourable case of equally sized coefficients, we could therefore identify dictionaries from signals with sparsity levels of the order $S\lesssim \mu^{-2}$, which, in case of incoherent dictionaries, means of the order of the ambient dimension $S \lesssim d$. \\ 
The second limitation of Theorem~\ref{th:exactasympt} is that, even if it allows for not exact S-sparseness of the signals, it does not take into account noise. 
Our next goal is therefore to extend the exact identification result in Theorem~\ref{th:exactasympt} to a stable identification result for less sparse (larger S) and noisy signals. For this task we first need to amend our signal model to incorporate noise. We would like to consider unbounded white noise, but also keep the property that in expectation the signals have unit square norm.
Further for the next section, where we want to transform our asymptotic identification results to results for finite sample sizes, it will be convenient if our signals are bounded. These considerations lead to the following model:
\begin{align}\label{noisymodel}
y=\frac{ \dico x +\noise}{\sqrt{1+\|\noise \|_2^2}},
\end{align}
where $\noise=(\noise(1) \ldots \noise(d))$ is a centred random subgaussian vector with parameter $\nsigma$. That is, the entries $\noise(i)$ are independent and satisfy $\E (e^{t \cdot\noise(i)}) \leq e^{t^2 \nsigma^2/2}$. \\
\noindent Employing this noisy signal model and formalising the ideas about the typical gap size between responses of the generating dictionary inside and outside the true support, leads to the following theorem.
\begin{theorem} \label{th:stableasympt}
Let $\dico$ be a unit norm frame with frame constants $A\leq B$ and coherence $\mu$.
Let the coefficients $x$ be drawn from a symmetric probability distribution $\nu$ on the unit sphere $S^{\natoms-1} \subset \R^\natoms$. Further let $\noise=(\noise(1) \ldots \noise(d))$ be a centred random subgaussian noise-vector with parameter $\nsigma$ and assume that the signals are generated according to the noisy signal model in~\eqref{noisymodel}.
If there exists $\beta>0$ such that
for $\amp_1(x)\geq \amp_2(x) \geq \ldots \geq \amp_K(x)\geq 0$ the non-increasing rearrangement of the absolute values of the components of $x$ we have
$ \amp_S(x)-\amp_{S+1}(x) \geq \beta $ almost surely 
and
\begin{align}\label{munoisecondcont}
\max\{\mu,\nsigma\} \leq \frac{\beta}{ \sqrt{ 72(\log a + \log\log a)} } \quad \mbox{for} \quad a= \frac{112 K^2 S (\sqrt{B}+1)}{C_\noise \beta (\bar c_1+ \ldots +\bar c_S)},
\end{align}
where $C_\noise= \E_\noise \left((1+\|\noise\|_2^2)^{-1/2}\right)$ and $\bar c_i := \E_x (\amp_i(x))$, then there is a local maximum of \eqref{expcrit} at $\tilde{\pdico}$ satisfying, 
\begin{align}
d(\tilde{\pdico},\dico)  \leq \frac{12 SK^2 \sqrt{B}}{C_\noise(\bar c_1+ \ldots +\bar c_S)}  \exp\left( \frac{-\beta^2}{72\max\{\mu^2,\nsigma^2\}} \right).
\end{align} 
\end{theorem}
\begin{proofsketch} As outlined at the beginning of the section the main ingredient we have to add to the proof idea of Theorem~\ref{th:exactasympt} is a probabilistic argument to substitute the condition guaranteeing that the S largest responses of the generating dictionary are $I_p$. Due to concentration of measure we get that for most sign sequences, and therefore most signals, the maximum is still attained at $I_p$. Moreover the gap to the remaining responses is actually large enough to accommodate relatively high levels of noise and/or perturbations. \\
The detailed proof can be found in Appendix~\ref{app:stableasympt}.
\end{proofsketch}
Let us make some observations about the last result.\\
First, we want to point out that for sub-gaussian noise with parameter $\nsigma$, the quantity $C_\noise= \E_\noise \left((1+\|\noise\|_2^2)^{-1/2}\right)$ in the statement above is well behaved. If for example the $r(i)$ are iid Bernoulli-variables, that is $P(r(i)= \pm \nsigma)=\frac{1}{2}$, we have $C_\noise=(1+d\nsigma^2)^{-1/2}$. In general we have the following estimate due for instance to Theorem~1 in \cite{hskazh11}. Since we have
\begin{align}
\P \left(\|x\|^2_2\geq \nsigma^2 (\ddim + 2 \sqrt{\ddim t } + 2t\right)\leq e^{-t},\notag
\end{align}
setting $t=d$, we get $ \P \left(\|x\|_2^2\geq 5\ddim \nsigma^2 \right)\leq e^{-d}$, which leads to
\begin{align}
\E_\noise \left(\frac{1}{\sqrt{1+\|\noise\|_2^2}} \right) \geq \frac{(1-e^{-d}) }{\sqrt{1+5\ddim \nsigma^2}}.\notag
\end{align}
Also to illustrate the result we again specialise it to the most favourable case of exactly $S$-sparse signals with balanced coefficients, that is $c_S(x)= S^{-1/2}$. Assuming white Gaussian noise with variance $\nsigma_G^2$ we see that identification is possible even for expected signal to noise ratios of the order $O(\frac{S}{d})$, that is 
$$\frac{\E (\|\dico x\|_2^2)}{\E(\|r\|_2^2)}\gtrsim \frac{S}{d}.$$
Similarly, by specialising Theorem~\ref{th:stableasympt} to the case of exactly $S$-sparse and noiseless signals we get - to the best of our knowledge - the first result establishing that locally it is possible to stably identify dictionaries from signals with sparsity levels beyond the spark of the generating dictionary. Indeed, even if some of the $S$-sparse signals could have representations in $\dico$ that require less than $S$ atoms, there will still be a local maximum of the asymptotic criterion close to the original dictionary as long as the smallest coefficient of each signal is of the order $O(\mu)$, which in the most favourable case means that we can have $S\lesssim \mu^{-2}$ or $S\lesssim d$. The quality of this result is on a par with the best results for finding sparse approximations in a \emph{given} dictionary, which say that on average Basis Pursuit or thresholding can find the correct sparse support even for signals with sparsity levels of the order of the ambient dimension \citep{tr08, scva07}.\\
Next note that with the available tools it would be possible to consider also a signal model where a small fraction of the coefficients violates the decay condition $\amp_S(x)-\amp_{S+1}(x) \geq \beta$ and still have stability. However, we leave explorations in that direction to the interested reader and instead turn to the study of the criterion for a finite number of training samples.

\section{Finite Sample Size Results \label{sec:finitesamples}}

In this section we will transform the two asymptotic results from the last section into results for finite sample sizes, 
that is, we will study when $\dico$ is close to a local maximum of
\begin{align}\label{finitecrit}
\max_{\pdico \in \mathcal D}\, \frac{1}{N} \sum_{n=1}^N \max_{|I| = S} \| \pdico_I^\star y_n\|_1,
\end{align} 
assuming that the $y_n$ are following either the noise-free or the noisy signal model.
For convenience we will do the analysis for the normalised version~\eqref{finitecrit} of the $S$-response criterion~\eqref{thecriterion}.
\begin{theorem}\label{th:exactfinite}
Let $\dico$ be a unit norm frame with frame constants $A\leq B$ and coherence $\mu$.
Let the coefficients $x_{n}$ be drawn from a symmetric probability distribution $\nu$ on the unit sphere $S^{\natoms-1} \subset \R^\natoms$ and assume that the signals are generated as $y_{n}=\dico x_{n}$.  If there exists $\beta>0$ such that
for $\amp_1(x_n)\geq \amp_2(x_n) \geq \ldots \geq \amp_K(x_n)\geq 0$ the non-increasing rearrangement of the absolute values of the components of $x_n$ we have
$ \amp_S(x_n)-\amp_{S+1}(x_n) - 2\mu \|x_n\|_1 \geq \beta $ almost surely
and the target precision $\tilde\eps$ satisfies
\begin{align}
\tilde \eps  \leq \frac{\beta}{1+ 3 \sqrt{ \log\left( \frac{50 K^2 S \sqrt{B}}{\beta (\bar c_1+ \ldots +\bar c_S)}\right) } },\notag
\end{align}
where $\bar c_i := \E_{x_n} (\amp_i(x_n))$, then except with probability,
\begin{align}
 2\exp\left(- \frac{N \tilde\eps^2 (\bar c_1+ \ldots +\bar c_S)^2}{129 S^2 K^2 B} + Kd \log  \left( \frac{25 SK \sqrt{B} }{\tilde\eps(\bar c_1+ \ldots +\bar c_S)} \right)\right),\notag
\end{align}
there is a local maximum of \eqref{finitecrit} respectively \eqref{thecriterion} at $\tilde{\pdico}$ satisfying, 
\begin{align}
d(\tilde{\pdico}, \dico) \leq  \tilde \eps +\frac{ \tilde \eps^2}{4K}.\notag
\end{align} 
\end{theorem}
\begin{theorem} \label{th:stablefinite}
Let $\dico$ be a unit norm frame with frame constants $A\leq B$ and coherence $\mu$.
Let the coefficients $x_{n}$ be drawn from a symmetric probability distribution $\nu$ on the unit sphere $S^{\natoms-1} \subset \R^\natoms$. Further let $\noise_{n}=(\noise_n(1) \ldots \noise_n(d))$ be i.i.d. centred random subgaussian noise-vectors with parameter $\nsigma$ and assume that the signals are generated according to the noisy signal model in~\eqref{noisymodel}.  If there exists $\beta>0$ such that
for $\amp_1(x_n)\geq \amp_2(x_n) \geq \ldots \geq \amp_K(x_n)\geq 0$ the non-increasing rearrangement of the absolute values of the components of $x$ we have
$ \amp_S(x_n)-\amp_{S+1}(x_n) \geq \beta $ almost surely 
and if the target precision $\tilde \eps$, the noiseparameter $\nsigma$ and the coherence $\mu$ satisfy
\begin{align}\label{stabletargetprec}
\tilde\eps&\leq \frac{\beta}{\frac{9}{4}+ 9\sqrt{ \log a } } 
\quad \mbox{ and }\\
\max\{\mu,\nsigma\}& \leq \frac{\beta}{ \sqrt{ 72(\log a + \log\log a)} } \quad \mbox{for} \quad a= \frac{150 K^2 S (\sqrt{B}+1)}{C_\noise \beta (\bar c_1+ \ldots +\bar c_S)},\notag
\end{align}
where $C_\noise= \E_{\noise_n} \left((1+\|\noise_n\|_2^2)^{-1/2}\right)$ and $\bar c_i := \E_{x_n} (\amp_i(x_n))$, then except with probability
\begin{align}
 &2\exp\left(- \frac{N \tilde \eps^2_{\mu,\nsigma}(\bar c_1+ \ldots +\bar c_S)^2}{513 C_\noise^2 S^2 K^2  \big(\sqrt{B}+1\big)^2} + Kd \log  \left( \frac{49 SK \big(\sqrt{B}+1\big) }{\eps_{\mu,\nsigma}(\bar c_1+ \ldots +\bar c_S)} \right)\right),\notag\\
\mbox{ where}\qquad &
\tilde \eps_{\mu,\nsigma}= \max \left\{ \tilde \eps, \frac{16  S K^2\big(\sqrt{B}+1\big)}{C_\noise(\bar c_1+ \ldots +\bar c_S)} \exp\left(- \frac{\beta^2}{72\max\{\mu^2,\nsigma^2\}} \right)\right\},\notag
\end{align} 
there is a local maximum of \eqref{finitecrit} respectively \eqref{thecriterion} at $\tilde{\pdico}$, satisfying
\begin{align}
d(\tilde{\pdico}, \dico)\leq \tilde \eps_{\mu,\nsigma} +\frac{\tilde \eps_{\mu,\nsigma}^2}{16K}.\notag
\end{align} 
\end{theorem}
\begin{proofsketch}The proofs, which can be found in Appendix~\ref{app:finite}, are based on three ingredients, a Lipschitz property for the mapping $\pdico \rightarrow \frac{1}{N} \sum_{n=1}^N \max_{|I| = S} \| \pdico_I^\star y_n\|_1$ for the respective signal model, the concentration of the sum around its expectation for a $\delta$-net covering the space of all admissible dictionaries close to $\dico$ and a triangle inequality argument to show that the finite sample response differences are close to the expected response differences and therefore larger than 0 for all $\eps\gtrsim \eps_{\mu,\nsigma}$.
\end{proofsketch}
To see better how the sample complexity behaves, we simplify the two theorems to the special case of noiseless exactly $S$-sparse signals with balanced coefficients for various orders of magnitude of $S$.\\
If we have $S= O(1)$, Theorem~\ref{th:exactfinite} implies that in order to have a maximum within radius $\tilde \eps$ to the original dictionary $\dico$ with probability $e^{-Kd}$ we need $N=O\big(K^3d\tilde\eps^{-2}\big)$ samples. Conversely given $N$ training signals we can expect the distance between generating dictionary and closest local maximum to be of the order $O\big(K^2N^{-1/2}\big)$.\\
If we assume a very incoherent dictionary where $\mu=O(d^{-1/2})$ and thus let the sparsity level be of the order $O(\sqrt{d})$ the sample complexity rises to $N=O\big(K^3d^{3/2} \tilde\eps^{-2}\big)$. Taking into account that by \eqref{stabletargetprec} the target precision $\tilde \eps$ needs to be of order $O\left(S^{-1/2}\right)=O\left(d^{-1/4}\right)$ this means that we need at least $N=O\big(K^3d^2\big)$ training signals and once this initial level is reached, the error goes to zero at rate $N^{-1/2}$. \\
For an even lower sparsity level, $S=O(d)$, again assuming a very incoherent dictionary, the sample complexity for target precision $\tilde \eps$ implied now by Theorem~\ref{th:stablefinite} rises to $N=O\big(K^3d^2 \tilde\eps^{-2}\big)$. In this regime, however, we cannot reach arbitrarily small errors by choosing $N$ large enough but only approach the asymptotic precision
$ \tilde \eps_\mu= 16K^2 \sqrt{SB}\exp\left(-d/72S \right)$.\\
Following these promising theoretical results, in the next section we will finally see how theory translates into practice.

\section{Experiments \label{sec:experiments}}

After showing that the optimisation criterion in \eqref{thecriterion} is locally suitable for dictionary identification, in this section we present an iterative thresholding and K means type algorithm (ITKM) to actually find the local maxima of \eqref{thecriterion} and conduct some experiments to illustrate the theoretical results.
We recall that given the input signals $Y=(y_1 \ldots y_N)$ and a fixed sparsity parameter $S$ we want to solve,
\begin{align}
\max_{\pdico \in \dicoset} \sum_n \max_{ |I| = S} \|\pdico_I^\star y_n\|_1.\notag
\end{align} 
Using LaGrange multipliers,
\begin{align*}
&\frac{\partial}{\partial \patom_k}\left(  \sum_n \max_{ |I| = S} \|\pdico_I^\star y_n\|_1 \right)= \sum_{n: k \in I(\pdico,y_n)} \signop(\ip{\patom_k}{y_n}) y_n^\star,\\
&\frac{\partial}{\partial \patom_k}\left(   \|\patom_k\|^2_2 \right)= 2 \patom_k^\star,
\end{align*}
where $I(\pdico,y_n):=\argmax_{ |I| = S} \|\dico_I^\star y_n\|_1$, we arrive at the following update rule,
\begin{align}\label{updaterule}
\patom_k^{new} = \lambda_k \cdot  \sum_{n: k \in I(\pdico^{old},y_n)} \signop(\ip{\patom^{old}_k}{y_n}) y_n,
\end{align}
where $\lambda_k$ is a scaling parameter ensuring that $\|\patom_k^{new}\|_2=1$.\\
In practice, when we do not have an oracle giving us the generating dictionary as initialisation, we also need to safeguard against bad initialisations resulting in a zero-update $\patom_k^{new}=0$. For example we can choose
the zero-updated atom uniformly at random from the unit sphere or from the input signals. \\
Note that finding the sets $ I(\pdico^{old},y_n)$ corresponds to $N$ thresholding operations while updating according to \eqref{updaterule} corresponds to K signed means. Altogether this means that each iteration of ITKM has computational complexity determined by the matrix multiplication $\pdico^\star Y$, meaning $O(dKN)$. This is light in comparison to K-SVD, which even when using thresholding instead of OMP as sparse approximation procedure still requires the calculation of the maximal singular vector of K on average $d\times \frac{N}{K}$ matrices. It is also more computationally efficient than the algorithm for local dictionary refinement, proposed in \cite{argemo13}, which is also based on averaging. Furthermore it is straightforward to derive online or parallelised versions of ITKM.  In an online version for each newly arriving signal $y_n$ we calculate $I(\pdico^{old},y_n)$ using thresholding and update $\patom_k^{new}=\patom_k^{new}+\signop(\ip{\patom^{old}_k}{y_n}) y_n$ for $k \in I(\pdico^{old},y_n)$. After $N$ signals have been processed we renormalise the atoms $\patom_k^{new}$ to have unit norm and set $\pdico^{old}=\pdico^{new}$. Similarly, to parallelise we divide the training samples into $m$ sets of size $\frac{N}{m}$. Then on each node $m$ we learn a dictionary $\pdico^{new}_m$ according to \eqref{updaterule} with $\lambda_k=1$. We then calculate the sum of these dictionaries $\pdico^{new}_0=\sum_m \pdico^{new}_m$ and renormalise the atoms in $\pdico^{new}_0$ to have unit norm.\\ 
Armed with this very simple algorithm we will now conduct four experiments to illustrate our theoretical findings\footnote{A Matlab penknife (mini-toolbox) for playing around with ITKM and reproducing the experiments can be found at \url{http://homepage.uibk.ac.at/~c7021041/ITKM.zip}}.

\begin{table}[htb]
\fbox{\parbox{\textwidth}{{\bf Signal Model}\\
Given the generating dictionary $\dico$ our signal model further depends on four coefficient parameters, 
\vspace{0.5em}

\begin{tabular}{rcl}
$S$ &-& the effective sparsity or number of comparatively large coefficients, \\
$b$ &-& deciding the decay factor of these sparse coefficients,\\
$T$ &-& the total number of non-zero coefficients ($T\geq S$) and \\
$\nsigma$ &-& the noise level.
\end{tabular}
\vspace{0.5em}

Given these parameters we 
choose a decay factor $c_b$ uniformly at random in the interval $[1-b,1]$. We set $c_i = c_b^i/\sqrt{S}$ for $1\leq i\leq S$ and $c_i=0$ for $T<i\leq K$. If $T=S$ we renormalise the sequence to have unit norm, while if $T > S$ we choose the vector $(c_{S+1}, \ldots, c_T)$ uniformly at random on the sphere of radius $R$, where $R$ is chosen such that the resulting sequence $c$ has unit norm.
We then choose a permutation $p$ and a sign sequence $\sigma$ uniformly at random 
and set $y=\dico c_{p,\sigma}$, respectively $y=(\dico c_{p,\sigma}+r)/\sqrt{1+\|r\|_2}$ where $r$ is a Gaussian noise-vector with variance $\nsigma^2$ if $\nsigma>0$.
}}\caption{Signal Model \label{coeffmodel}}
\end{table}

\subsection{ITKM vs. K-SVD}
In our first experiment we compare the local recovery error of ITKM and K-SVD for 3-dimensional bases with increasing condition numbers.\\
The bases are perturbations of the canonical basis $\dico=(e_1,e_2,e_3)$ with the vector $v=(1,1,1)$. That is, $\dico^t=(e^t_1,e^t_2,e^t_3)$, where $e^t_i=(e_i+t v)/\|(e_i+t v)\|_2$ and $t$ varying from 0 to 0.5 in steps of 0.1, which corresponds to condition numbers $\kappa(\dico^t)$ varying from 1 to 2.5. We generate $N=4096$ approximately $1$-sparse noiseless signals from the signal model described in Table~\ref{coeffmodel} with $S=1$, $T=2$, $\nsigma=0$ and $b=0.1/0.2$ and run both ITKM and K-SVD with 1000 iterations, sparsity parameter $S=1$ and the true dictionary (basis) as initialisation. Figure~\ref{fig1}(a) shows the recovery error $d(\dico^t,\tilde{\pdico})$ between the original dictionary and the output of the respective algorithm averaged over 10 runs. \\
As predicted by the theoretical results on the corresponding underlying minimisation principles, the recovery errors of ITKM and K-SVD are roughly the same for $\dico^0$, which is an orthogonal basis and therefore tight. However, while for ITKM the recovery error stays constantly low over all condition numbers, for K-SVD it increases with increasing condition number or non-tightness.

\subsection{Recovery Error and Sample Size}
The next experiment is designed to show how fast the maximiser $\tilde \pdico$ near the original dictionary $\dico$ converges to $\dico$ with increasing sample size $N$.\\
The generating dictionaries consist of the canonical basis in $\R^d$ for $d=4,8,16$ and the first $d/2$ elements of the Hadamard basis and as such are not tight. For every set of parameters $d, S(T), b$ we generate $N$ noiseless signals with $N$ varying from $2^7=128$ to $2^{14}=16 384$ and run ITKM with 1000 iterations, sparsity parameter $S$ equal to the coefficient parameter $S$ and the true dictionary as initialisation. Figure~\ref{fig1}(b) shows the recovery error $d(\dico,\tilde{\pdico})$ between the original dictionary $\dico$ and the output of ITKM $\tilde{\pdico}$ averaged over 10 runs.\\
As predicted by Theorem~\ref{th:exactfinite} the recovery error decays as $N^{-1/2}$. However, the separation of the curves for $d=4,8,16$ and almost exactly sparse signals $(b=0.01)$ by a factor around $\sqrt{2}$ instead of $4$, as suggested by the estimate $\tilde \eps \approx K^2N^{-1/2}$, indicates that the cubic dependence of the sampling complexity on the number of atoms $K$ may be too pessimistic and could be lowered.

 \begin{figure}[thb]
\hspace{-1cm}
\begin{tabular}{cc}
 \includegraphics[width=8cm]{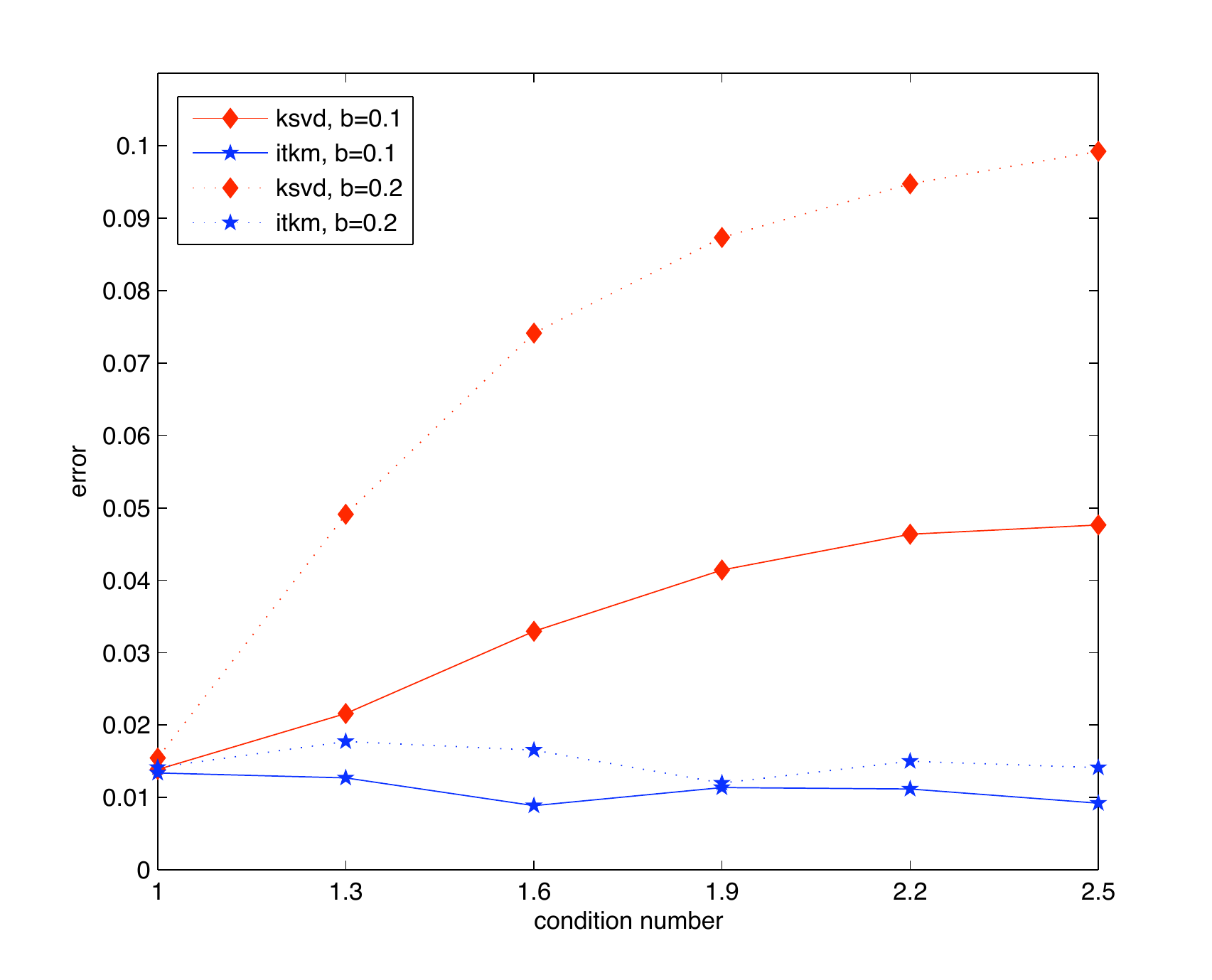} & \includegraphics[width=8cm]{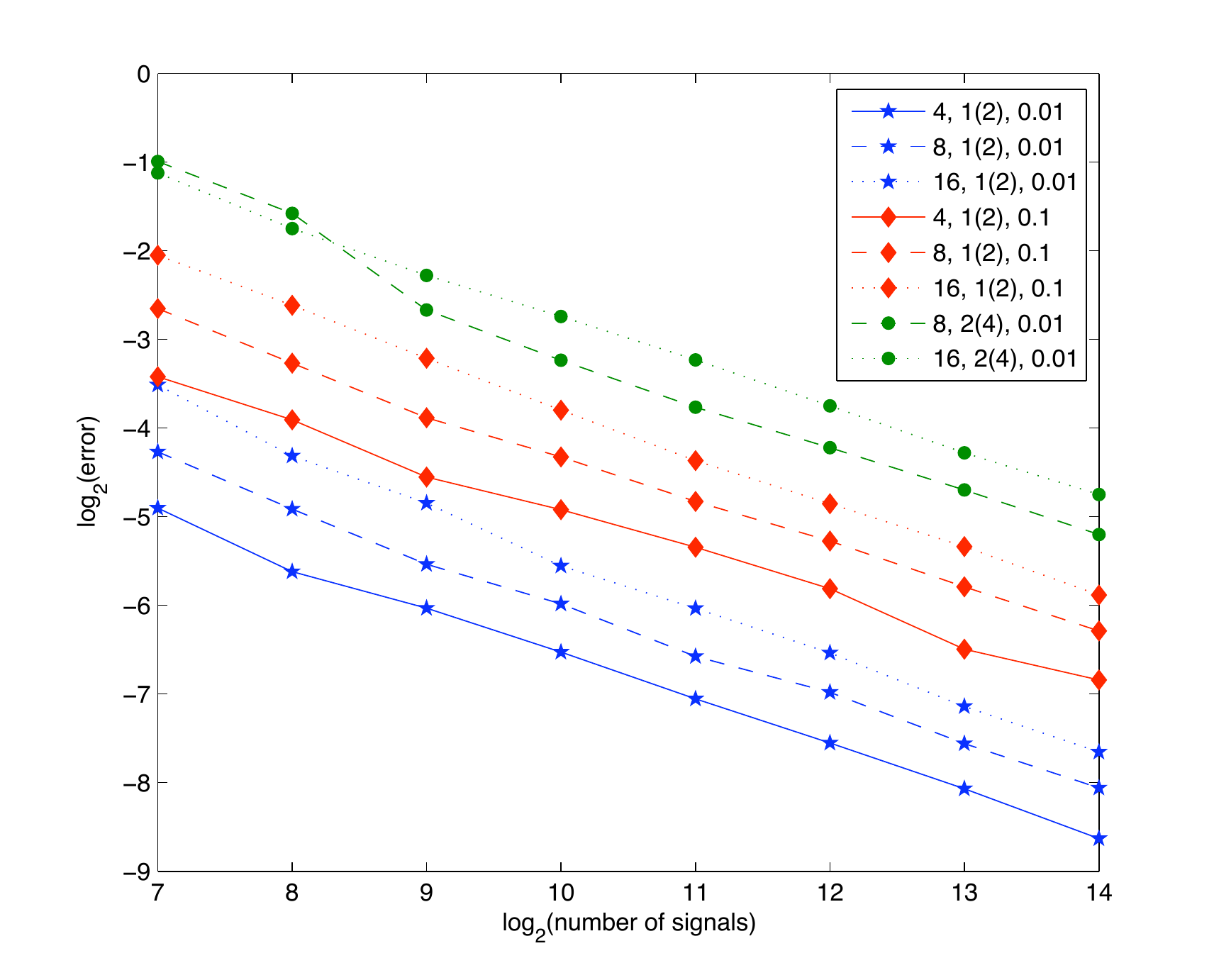}\\
  (a) & (b)
  \end{tabular}
    \caption{(a) Local recovery error of K-SVD and ITKM for two different types of decaying coefficients and
    bases with varying condition numbers in $\R^3$, (b) Decay of recovery error of ITKM with increasing number of training signals \label{fig1}}
\end{figure}

\subsection{Stability of Recovery Error under Coherence and Noise}
With the last two experiments we illustrate the stability of the maximisation criterion under coherence and noise. As generating dictionaries we use again the canonical basis plus half Hadamard dictionaries described in the last experiment, which have coherence $\mu=d^{-1/2}$.To test the stability under coherence we use a large enough number of noiseless training signals $N=16 384$, such that the distance between the local maximum of the criterion near the generating dictionary, that is the output of ITKM with oracle initialisation, and the generating dictionary is mainly determined by the ratio between the gap size $\beta$ and the coherence. For each set of parameters $d, S(T)$ we create $N$ training signals with decreasing gap sizes $\beta$ by increasing $b$ from 0 to 0.1 in steps of 0.01 and run ITKM with oracle initialisation, parameter $S$ and 1000 iterations. Figure~\ref{fig2}(a) shows the recovery error $d(\dico,\tilde{\pdico})$ between the original dictionary $\dico$ and the output of ITKM $\tilde{\pdico}$ again averaged over 10 trials. \\
Again the experiments reflect our theoretical results. For $d=8,16$ with $S=1$ or $d=16$ with $S=2$ the gap size is large enough that over the whole range of parameters the recovery error stays constantly low at the level defined by the number of samples. Note that this is quite good, since for $b=0.1$ we are already far beyond the gap size coherence ratios where the stable theoretical results hold. On the other hand for $d=8$ with $S=2$ or $d=16$ with $S=3$ early on the gap decreases enough to become the error determining factor and so we see an increase in recovery error as $b$ grows. \\
Conversely to test the stability under noise we use a large enough number of exactly sparse training signals, such that the recovery error will be mainly determined by the noise level. For each set of parameters $d, S(S), b$ we create $N=16384$ training signals with Gaussian noise of variance (noise level) $\nsigma^2$ going from $0$ to $0.1$ in steps of $0.01$ and run ITKM with oracle initialisation, parameter $S$ and 1000 iterations. Figure~\ref{fig2}(b) shows the recovery error $d(\dico,\tilde{\pdico})$ between the original dictionary $\dico$ and the output of ITKM $\tilde{\pdico}$, this time averaged over 20 trials. \\
\begin{figure}[thb]
\hspace{-1cm}
\begin{tabular}{cc}
   \includegraphics[width=8cm]{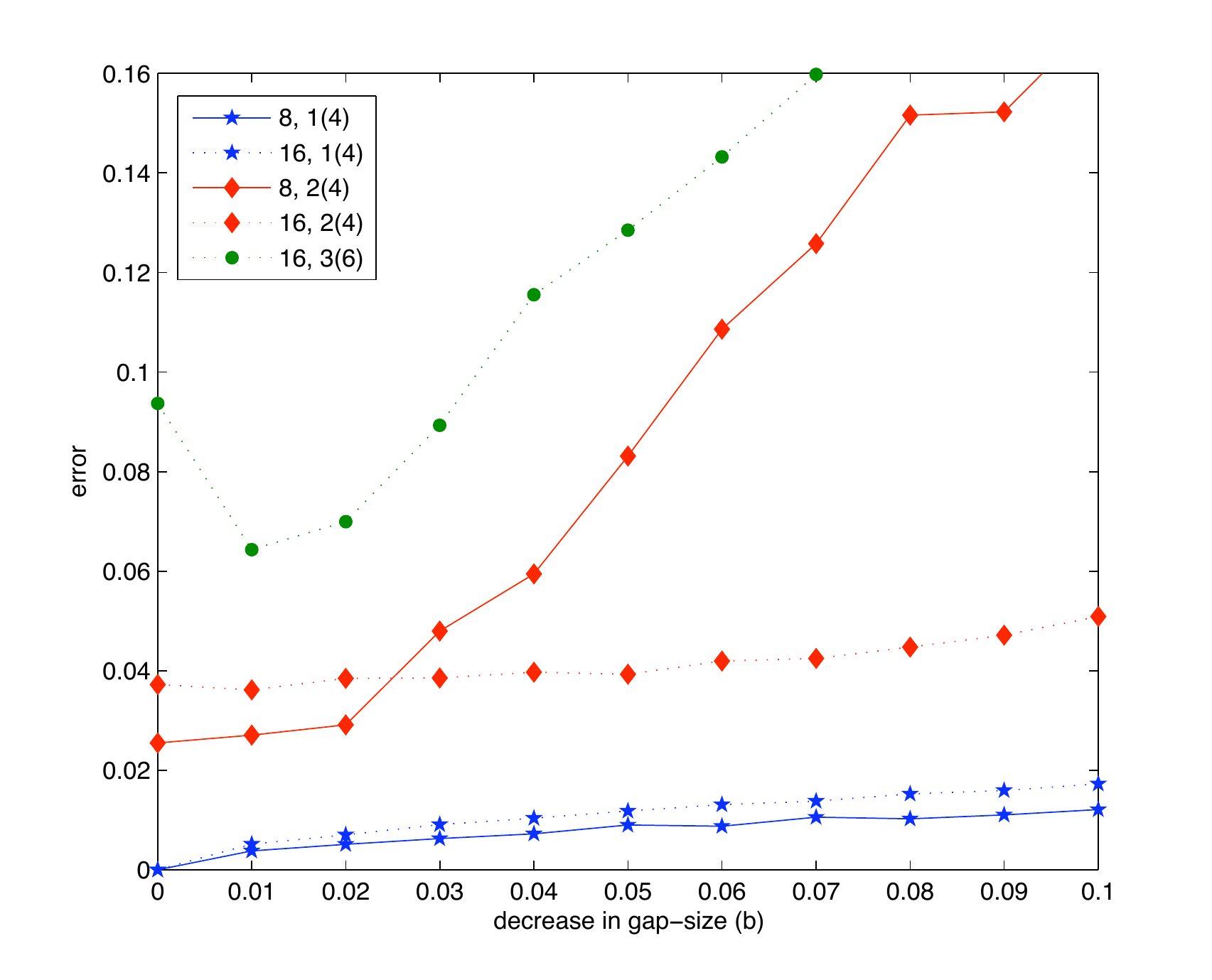} & \includegraphics[width=8cm]{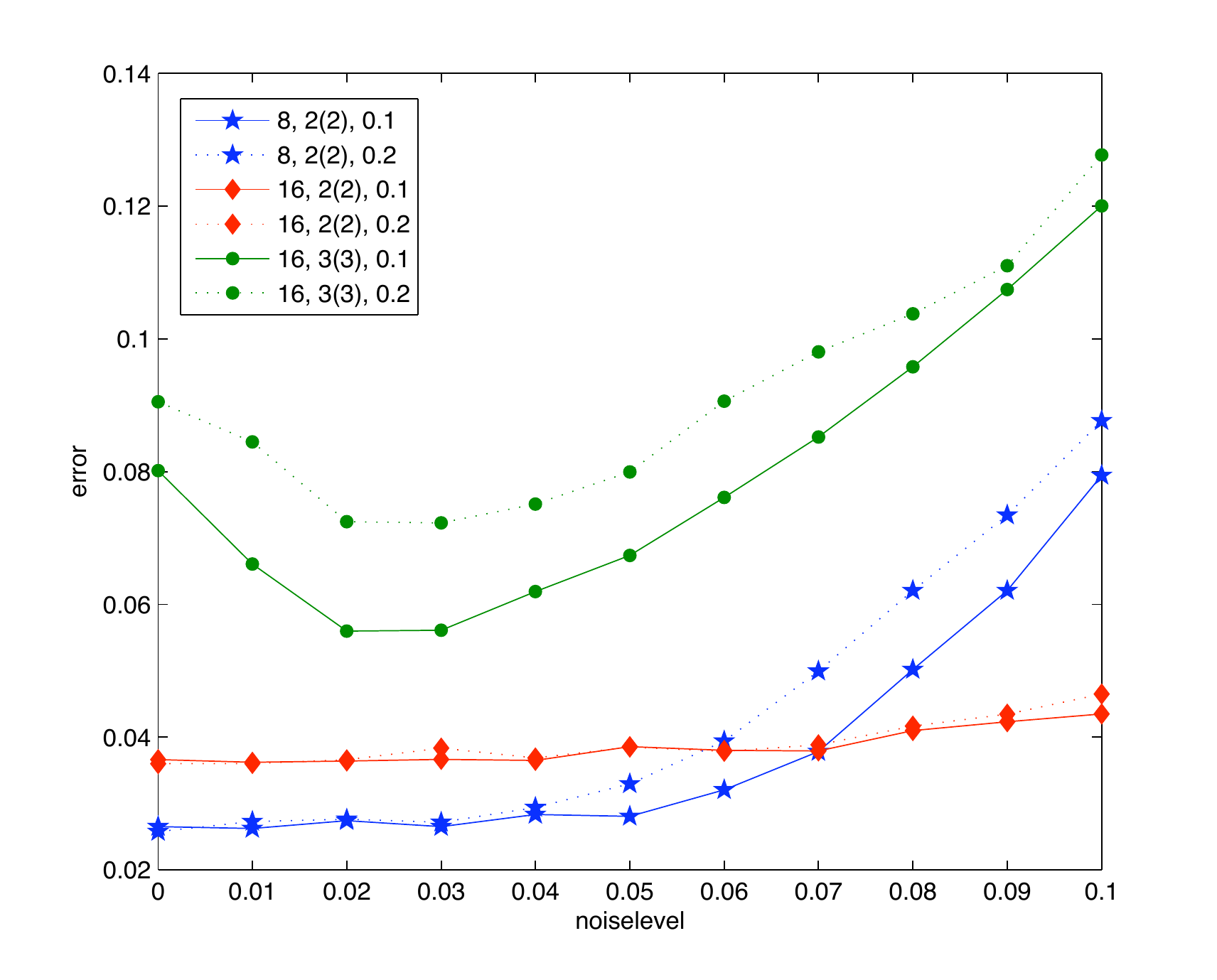}\\
  (a) & (b)
  \end{tabular}
    \caption{Increase of recovery error with (a) decreasing ratio between coefficient gap and coherence and (b) increasing noise level \label{fig2}}
\end{figure}
The curves again correspond to the prediction of the theoretical results, that is the recovery error stays at roughly the same level defined by the number of samples until the noise becomes large enough and then increases.
What is maybe interesting to observe in both experiments is the dithering effect for $d=16$ with $S=3$, which is due to the special structure of the dictionary. Indeed using almost equally sized, almost exactly sparse coefficients, it is possible to build signals using only the canonical basis that have almost the same response in only half the Hadamard basis and the other way round. This indicates that slight perturbations of one with the other lead to even better responses and therefore a larger recovery error.
After showing that the theoretical results translate into algorithmic practice, we finally turn to a discussion of our results in the context of existing work and point out directions of future research.

\section{Discussion \label{sec:discussion}}
We have introduced a new response maximisation principle for dictionary learning and shown that this is locally suitable to identify a generating $\mu$-coherent dictionary from approximately $S$-sparse training signals to arbitrary precision as long as the sparsity level is of the order $O(\mu^{-1})$. We have also presented - to the best of our knowledge - the first results showing that stable dictionary identification is locally possible not only for signal to noise ratios of the order $O(\sqrt{d})$ but also for sparsity levels of the order $O(\mu^{-2})$. \\
The derived sample complexity (omitting log factors) of $O(K^3d \tilde \eps^{-2})$, for signals with sparsity levels $S=O(1)$ is roughly the same as for the K-SVD criterion, \cite{sc14}, or the $\ell_1$-minimisation criterion, \cite{bagrje14}, but somewhat large compared to recently developed dictionary algorithms that have a sample complexity of $O(K^2)$, \cite{argemo13, aganjaneta13}, or $O(K\eps^{-2})$, \cite{aganne13}. However, as the sparsity approaches and goes beyond $\mu^{-1}\sim \sqrt{d}$ the derived sample complexity of $O(K^3d^2 \tilde \eps^{-2})$ compares quite favourably to the sample complexity of $O(K^{1/(4\eta)})$ for a sparsity level $d^{1/2-\eta}$ as projected in \cite{argemo13}. Given that also our experimental results suggest that $O(K^3d \tilde \eps^{-2})$ is quite pessimistic, one future direction of research aims to lower the sample complexity. In particular ongoing work suggests that for the ITKM \emph{algorithm} a sample size of order $K^2$ is enough to guarantee local recovery with high probability.\\
Another strong point of the results is that the corresponding maximisation algorithm ITKM (Iterative Thresholding and $K$ signed Means) is locally successful, as demonstrated in several experiments, and computationally very efficient. The most complex step in each iteration is the matrix multiplication $\dico^\star Y$ of order $O(dKN)$, which is even lighter than the iterative averaging algorithm described in \cite{argemo13}.\\
However, the serious drawback is that ITKM is only a local algorithm and that all our results are only local. Also while for the K-SVD criterion and the $\ell_1$-minimisation criterion there is reason to believe that all local minima might be equivalent, the response maximisation principle has a lot of smaller local maxima, which is confirmed by preliminary experiments with random initialisations. There ITKM fails but with grace, that is, it outputs local maximisers that have not all, but only most atoms in common with what seems to be the global maximiser near the generating dictionary.
This behaviour is in strong contrast to the algorithms presented in \cite{argemo13, aganne13}, that have global success guarantees at a computational cost of the order $O(dN^2)$, and leads to several very important research directions.\\
First we want to confirm that ITKM has a convergence radius of the order $O(1/\sqrt{S})$. This is suggested by the derived radius of the area on which the generating dictionary is the optimal maximiser as well as preliminary experiments. Alternatively, we could investigate how the results for the local iterative algorithms in \cite{argemo13, aganjaneta13} could be extended to larger sparsity levels and convergence radii using our techniques. The associated important question is how to extend the results for the algorithms presented in \cite{argemo13, aganne13} to sparsity levels $O(\mu^{-2})$, if possible at lower cost than $O(dN^2)$. Given the conjectured size of the convergence radius for ITKM it would even be sufficient for the output of the algorithm to arrive at a dictionary within distance $O(1/\sqrt{S})$ to the generating dictionary, since the output could then be used as initialisation for ITKM. \\
A parallel approach for getting global identification results for sparsity levels $O(\mu^{-2})$, that we are currently pursuing, is to analyse a version of ITKM using residual instead of pure signal means, which in preliminary experiments exhibits global convergence properties.\\
The last research directions we want to point out are concerned with the realism of the signal model. The fact that for an input sparsity $S$ a gap of order $O(\mu^{-2})$ between the $S$ and $S+1$ largest coefficient is sufficient can be interpreted as a relaxed dependence of the algorithm on the sparsity parameter, since a gap of order $\mu^{-2}$ can occur quite frequently. To further decrease this sensitivity to the sparsity parameter in the criterion and the algorithm we would therefore like to extend our results to the case where we can only guarantee a gap of order $O(\mu^{-2})$ between the $S$ largest and the $S+T$ largest coefficient for some $T>1$. Last but not least we would like to exactly reflect the practical situation, where we would normalise our training signals to equally weight their importance and analyse the unit norm signal model where $y=\dico x + r/\|\dico x +r\|_2$.

\acks{This work was supported by the Austrian Science Fund (FWF) under Grant no. J3335
and improved thanks to the reviewers' comments and suggestions. 
Thanks also go to the Computer Vision Laboratory of the University of Sassari, Italy, which provided the surroundings, where almost all of the presented work was done, and to Michael Dymond, who looked over the final manuscript with a native speaker's eye.}

\appendix

\section{Proofs}

\subsection{Proof of Theorem~\ref{th:exactasympt}}\label{app:exactsympt}
We first reformulate and prove the theorem for the simple case of a symmetric coefficient distribution based on one sequence and then use an integration argument to extend it to the general case.
\begin{proposition}\label{th:exactsimple}
Let $\dico$ be a unit norm frame with frame constants $A\leq B$ and coherence $\mu$.
Let $x \in \R^K$ be a random permutation of a sequence $\amp$, where $\amp_1 \geq \amp_2 \geq \amp_3 \ldots \geq \amp_\natoms \geq 0$ and $\|\amp\|_2=1$, provided with random $\pm$ signs, that is $x=\amp_{p, \sigma}$ with probability $\P(p,\sigma)=(2^\natoms \natoms!)^{-1}$. Assume that the signals are generated as $y=\dico x$.
If $\amp$ satisfies $\amp_S > \amp_{S+1} + 2 \mu \|\amp\|_1$ then there is a local maximum of \eqref{expcrit} at $\dico$. \\
Moreover for $\pdico \neq \dico$ we have $\E_y \left( \max_{|I| = S} \| \pdico_I^\star y\|_1 \right) < \E_y \left( \max_{|I| = S} \| \dico_I^\star y\|_1 \right) $ as soon as 
\begin{align}\label{epscondexact}
d(\dico,\pdico) \leq \frac{c_S-c_{S+1}- 2\mu\|c\|_1}{1 + 3\sqrt{ \log\left( \frac{25 K^2 S \sqrt{B}}{(c_S-c_{S+1}- 2\mu\|c\|_1)(c_1+ \ldots +c_S)}\right) } }.
 \end{align}
\end{proposition}
\begin{proof}
We start by evaluating the objective function at the original dictionary~$\dico$. 
\begin{align}
 \E_y \left( \max_{|I| = S} \| \dico_I^\star y\|_1 \right) &= \E_p \E_\sigma \left( \max_{|I| = S} \| \dico_I^\star \dico \amp_{p,\sigma}\|_1 \right)  =\E_p \E_\sigma \left( \max_{|I| = S} \sum_{i\in I } |\ip{ \atom_i}{ \dico \amp_{p,\sigma}}| \right).\notag
\end{align} 
To estimate the sum of the (in absolute value) largest $S$ inner products, we first assume that $p$ is fixed. Setting $I_p = p^{-1}\left(\{1,\ldots S\}\right)$ we have, 
\begin{align}  
 |\ip{ \atom_i}{\dico \amp_{p,\sigma}}|&= \Big|  \sigma_{i}  \amp_{p(i)} + \sum_{j\neq i} \sigma_j \amp_{p(j)}\ip{ \atom_{i}}{\atom_j  } \Big| 
\begin{array}{lr}
 \geq c_S - \mu \|\amp\|_1 & \forall i \in I_p\\
  \leq c_{S+1} + \mu \|\amp\|_1 & \forall i \notin I_p
  \end{array}.\notag
\end{align} 
Together with the condition that $\amp_S > \amp_{S+1} + 2 \mu \|\amp\|_1$ these estimates ensure that the $S$ maximal inner products in absolute value are attained at $I_p$
and so we get for the expectation,
\begin{align}
\E_p \E_\sigma \left( \max_{|I| = S} \| \dico_I^\star \dico \amp_{p,\sigma}\|_1 \right) &= 
\E_p \E_\sigma \left(\| \dico_{I_p}^\star \dico \amp_{p,\sigma}\|_1 \right) \notag \\
&=\E_p \E_\sigma \left(\sum_{i\in I_p} \Big|   \amp_{p(i)} +  \sigma_{i} \sum_{j\neq i} \sigma_j \amp_{p(j)}\langle \atom_{i}, \atom_j  \rangle \Big| \right) =c_1 +\ldots +c_S.\notag
\end{align} 
%
To compute the expectation for a perturbation of the original dictionary we use the following parametrisation of all $\eps$-perturbations $\pdico$ of the original dictionary $\dico$. If $d(\pdico,\dico)=\eps$ then $\|\patom_i - \atom_i\|_2=\eps_i$ with $\max_i \eps_i = \eps$ and we have $z_i$ with $\langle \atom_i,z_i\rangle = 0, \|z_i\|_2=1$ and 
$\alpha_i := 1-\eps^2_i/2$ and $\omega_i := (\eps_i^2 - \eps_i^4/4)^{\frac{1}{2}}$, such that
\begin{align}
\patom_i = \alpha_i \atom_i + \omega_i z_i.\notag
\end{align}
Expanding the expectation as before we get,
 \begin{align}
 \E_y \left( \max_{|I| = S} \| \pdico_I^\star y\|_1 \right) &= \E_p \E_\sigma \left( \max_{|I| = S} \| \pdico_I^\star \dico \amp_{p,\sigma}\|_1 \right) = \E_p \E_\sigma \left( \max_{|I| = S} \sum_{i\in I } |\ip{ \patom_i}{ \dico \amp_{p,\sigma}}| \right). \label{exp:pert}
\end{align} 
The tried and tested strategy applied now is showing that for small perturbations and most sign patterns $\sigma$ the maximal inner products are still attained by $i \in I_p$. 
We have 
\begin{align}  
 \forall i \in I_p: \quad |\langle \patom_{i}, \dico \amp_{p,\sigma} \rangle|
  &\geq \alpha_{i} (c_S -  \mu \|\amp\|_1) - \omega_{i} | \ip{ z_{i}}{ \dico \amp_{p,\sigma} }| \notag\\
   \forall i \notin I_p: \quad |\langle \patom_{i}, \dico \amp_{p,\sigma} \rangle|
  &\leq \alpha_{i} (c_{S+1} + \mu \|\amp\|_1) + \omega_i  | \ip{ z_{i}}{ \dico \amp_{p,\sigma} }| \notag.
 \end{align} 
 %
Using Hoeffding's inequality we can estimate the typical sizes of the terms  $ | \ip{ z_{i}}{ \dico \amp_{p,\sigma} }|$,
\begin{align}
\P( | \langle z_{i}, \dico \amp_{p,\sigma} \rangle| \geq t) &= \P( | \sum_{j\neq i} \sigma_j \amp_{p(j)}\langle z_{i}, \atom_j  \rangle| >t)\notag \\
&\leq 2 \exp \left( -\frac{t^2}{2 \sum_{j\neq i}\amp_{p(j)}^2\langle z_{i}, \atom_j  \rangle^2}\right) \leq 2 \exp \left( -\frac{t^2}{2 }\right).\notag
\end{align}
In case $\omega_{i} \neq 0$ or equivalently $\eps_{i}\neq 0$, we set $t=s/\omega_{i}$ to arrive at 
\begin{align}
\P(\omega_{i} | \langle z_{i}, \dico \amp_{p,\sigma} \rangle|  \geq s ) \leq 2 \exp \left(-\frac{s^2}{2\omega_{i}^2 }\right) \leq 2 \exp \left(-\frac{s^2}{2 \eps_{i}^2}\right),\notag
\end{align}
where we have used that $\omega_{i}^2= \eps_{i}^2 - \eps_{i}^4/4 \leq \eps_{i}^2$, while in case $\eps_i = 0$ we trivially have that $\P( \omega_i |\ip{z_i}{\dico \amp_{p,\sigma}}| \geq s) = 0.$
Summarising these findings we see that except with probability 
\begin{align}
\eta :=2\sum_{i:\eps_i\neq 0} \exp \left(-\frac{s^2}{2 \eps_{i}^2}\right),\notag
\end{align}
we have
\begin{align}  
 \forall i \in I_p: \quad |\langle \patom_{i}, \dico \amp_{p,\sigma} \rangle|
  &\geq \alpha_{i} (c_S -  \mu \|\amp\|_1) - s \notag\\
   \forall i \notin I_p: \quad |\langle \patom_{i}, \dico \amp_{p,\sigma} \rangle|
  &\leq \alpha_{i} (c_{S+1} + \mu \|\amp\|_1) + s.\notag
 \end{align}
This means that as long as $\min_{ i \in I_p}  \alpha_{i} (c_S -  \mu \|\amp\|_1) - s \geq \max_{ i \notin I_p}   \alpha_{i} (c_{S+1} + \mu \|\amp\|_1) + s$, which is for instance implied by setting
$
s:= \frac{1}{2}(c_S-c_{S+1}- 2\mu\|c\|_1 - \frac{\eps^2}{2})
$, we have
\begin{align}
  \max_{|I| = S} \| \pdico_I^\star \dico \amp_{p,\sigma}\|_1
  =\| \pdico_{I_p}^\star \dico \amp_{p,\sigma}\|_1.\notag
\end{align}
We now use this result for the calculation of the expectation over $\sigma$ in \eqref{exp:pert}. For any permutation $p$ we define the set
\begin{align}
\Sigma_p:=\bigcup_{i} \{\sigma \mbox{ s.t. }  \omega_i|\ip{z_i}{\dico \amp_{p,\sigma}}| \geq  s\}.\notag
\end{align}
We then have
\begin{align}
\E_\sigma \left(  \max_{|I| = S} \| \pdico_I^\star \dico \amp_{p,\sigma}\|_1 \right)&=
 \sum_{\sigma\in \Sigma_p} \P(\sigma) \cdot  \max_{|I| = S} \| \pdico_I^\star \dico \amp_{p,\sigma}\|_1
 + \sum_{\sigma\notin \Sigma_p}\P(\sigma) \cdot \| \pdico_{I_p}^\star \dico \amp_{p,\sigma}\|_1 \notag \\
 &\hspace{-2cm}=
 \sum_{\sigma\in \Sigma_p} \P(\sigma) \cdot \left( \max_{|I| = S} \| \pdico_I^\star \dico \amp_{p,\sigma}\|_1 -\| \pdico_{I_p}^\star \dico \amp_{p,\sigma}\|_1\right) +\E_\sigma\left( \| \pdico_{I_p}^\star \dico \amp_{p,\sigma}\|_1\right). \label{pertexp1}
\end{align}
To estimate the sum over $\Sigma_p$, note that we have the following bounds:
\begin{align}
|\ip{\patom_i}{\dico \amp_{p,\sigma}}| &= |\alpha_i\ip{ \atom_i}{\dico \amp_{p,\sigma}} + \omega_i \ip{z_i}{\dico \amp_{p,\sigma}}|
\left\{ \begin{array}{ll}
\leq ( 1- \frac{\eps^2}{2}) |\ip{ \atom_i}{\dico \amp_{p,\sigma}}| + \eps \sqrt{B}\\
\geq ( 1- \frac{\eps^2}{2}) |\ip{ \atom_i}{\dico \amp_{p,\sigma}}| - \eps \sqrt{B}\\
\end{array}\right. ,\notag
\end{align}
leading to
\begin{align}
\max_{|I| = S} \| \pdico_I^\star \dico \amp_{p,\sigma}\|_1 &\leq ( 1- \frac{\eps^2}{2})\max_{|I| = S} \| \dico_I^\star \dico \amp_{p,\sigma}\|_1 + S \cdot \eps \sqrt{B}=( 1- \frac{\eps^2}{2})\| \dico_{I_p}^\star \dico \amp_{p,\sigma}\|_1 + S \cdot \eps \sqrt{B}\notag\\
\| \pdico_{I_p}^\star \dico \amp_{p,\sigma}\|_1 &\geq ( 1- \frac{\eps^2}{2}) \| \dico_{I_p}^\star \dico \amp_{p,\sigma}\|_1 - S \cdot \eps \sqrt{B}.\notag
\end{align}
Substituting these estimates into \eqref{pertexp1} we get
\begin{align}
\E_\sigma \left(  \max_{|I| = S} \| \pdico_I^\star \dico \amp_{p,\sigma}\|_1 \right) &\leq  \sum_{\sigma\in \Sigma_p} \P(\sigma) \cdot 2\eps S \sqrt{B} +\E_\sigma\left( \| \pdico_{I_p}^\star \dico \amp_{p,\sigma}\|_1\right)\notag\\
& \leq \eta \cdot 2\eps S \sqrt{B} +\E_\sigma\left( \| \pdico_{I_p}^\star \dico \amp_{p,\sigma}\|_1\right).\notag
 \end{align}
Next we calculate $\E_\sigma\left( \| \pdico_{I_p}^\star \dico \amp_{p,\sigma}\|_1\right)$:
\begin{align}
\E_\sigma\left( \| \pdico_{I_p}^\star \dico \amp_{p,\sigma}\|_1\right)&=
\E_\sigma\left( \sum_{i \in I_p} | \ip{\patom_i}{ \dico \amp_{p,\sigma}}|\right)\notag \\
& =\E_\sigma\left( \sum_{i \in I_p} \Big | \alpha_i c_{p(i)} + \sigma_i \ip{ \alpha_i \atom_i+\omega_i z_i}{ \sum_{j\neq i}  \sigma_j c_{p(j)}\atom_j} \Big | \right) =\sum_{i \in I_p} \alpha_i c_{p(i)},\label{absperturbation}
 \end{align}
 where have used that $\eps \leq ( (1-\frac{\eps^2}{2})c_S - \mu\|c\|_1)/\sqrt{B}$ guarantees the expression within absolute values in \eqref{absperturbation} to always be positive. 
Collecting all these results we arrive at the following estimate for the value of the objective function at $\pdico$:
 \begin{align*}
 \E_y \left( \max_{|I| = S} \| \pdico_I^\star y\|_1 \right) &= \E_p \E_\sigma \left( \max_{|I| = S} \| \pdico_I^\star \dico \amp_{p,\sigma}\|_1 \right) \\
 &\hspace{-1cm}\leq  \E_p \left(  4\eps S \sqrt{B}\sum_{i:\eps_i\neq 0} \exp \left(-\frac{(c_S-c_{S+1}- 2\mu\|c\|_1 - \frac{\eps^2}{2})^2}{8 \eps_{i}^2 }\right) + \sum_{i \in I_p} \alpha_i c_{p(i)} \right) \\
  &\hspace{-1cm}\leq4\eps S \sqrt{B}\sum_{i:\eps_i\neq 0} \exp \left(-\frac{(c_S-c_{S+1}- 2\mu\|c\|_1 - \frac{\eps^2}{2})^2}{8 \eps_{i}^2}\right) +\frac{c_1+ \ldots +c_S}{K} \sum_{i} \alpha_i .
\end{align*} 
Finally we are able to compare the expectation at the original dictionary to that at an $\eps$-perturbation. Remembering that $\alpha_i = 1- \frac{\eps_i^2}{2}$, we get
\begin{align}
 \E_y \left( \max_{|I| = S} \| \dico_I^\star y\|_1 \right)& -  \E_y \left( \max_{|I| = S} \| \pdico_I^\star y\|_1 \right)\notag\\
 &\hspace{-1cm}\geq \frac{c_1+ \ldots +c_S}{K} \sum_{i}\frac{\eps_i^2}{2} -4\eps S \sqrt{B}\sum_{i:\eps_i\neq 0} \exp \left(-\frac{(c_S-c_{S+1}- 2\mu\|c\|_1 - \frac{\eps^2}{2})^2}{8 \eps_{i}^2 }\right)\notag \\
   &\hspace{-1cm} \geq  \eps^2 \frac{c_1+ \ldots +c_S}{2K}-  4\eps S K \sqrt{B}  \exp \left(-\frac{(c_S-c_{S+1}- 2\mu\|c\|_1 - \frac{\eps^2}{2})^2}{8 \eps^2 }\right) .\notag
\end{align}
Thus to have a local maximum at the original dictionary we need that 
\begin{align}
 \eps > \frac{8  S K^2 \sqrt{B}}{c_1+ \ldots +c_S} \exp \left(-\frac{(c_S-c_{S+1}- 2\mu\|c\|_1 - \frac{\eps^2}{2})^2}{8 \eps^2 }\right),\notag
\end{align}
and all that remains to be shown is that this is implied by \eqref{epscondexact}.
Since $K\geq 2$, \eqref{epscondexact} implies that $\frac{\eps^2}{2} < \frac{c_S-c_{S+1}- 2\mu\|c\|_1}{2(1 + 3\sqrt{ \log96})^2} \leq \frac{c_S-c_{S+1}- 2\mu\|c\|_1}{100}$ and it suffices to show that \eqref{epscondexact} further implies 
\begin{align}\label{toshow}
 \eps> \frac{8  S K^2 \sqrt{B}}{c_1+ \ldots +c_S} \exp \left(-\frac{(c_S-c_{S+1}- 2\mu\|c\|_1)^2 \cdot 99^2}{8  \eps^2 \cdot 100^2 }\right).
\end{align}
Applying Lemma A.3 from \cite{sc14}, which says that for $a,b,\eps > 0$,
\begin{align}
\eps \leq \frac{4b}{1+ \sqrt{1+16\log(\frac{a}{b})}} \quad \mbox{implies that} \quad a \exp \left(\frac{-b^2}{\eps^2} \right)< \eps,\notag
\end{align} 
to the situation at hand, where $a= \frac{8  S K^2 \sqrt{B}}{c_1+ \ldots +c_S} $ and $b= \frac{(c_S-c_{S+1}- 2\mu\|c\|_1) \cdot 99}{\sqrt{8} \cdot 100}$, we get that \eqref{toshow} is ensured by
\begin{align}
\eps <\frac{c_S-c_{S+1}- 2\mu\|c\|_1}{ \sqrt{8} \cdot \frac{25}{99} \left(1 +  \sqrt{16\log\left( \frac{8\sqrt{8}\cdot \frac{100}{99}  e^{1/16} S K^2 \sqrt{B}}{(c_S-c_{S+1}- 2\mu\|c\|_1)(c_1+ \ldots +c_S)}\right)}\right)},\notag
\end{align}
which simplifies to \eqref{epscondexact}.
\end{proof}
\begin{proof}[of Theorem~\ref{th:exactasympt}]\\
Using the symmetry of $\nu$, our strategy is to reduce the general to the simple coefficient model.
Let $\amp$ denote the mapping that assigns to each $x\in S^{K-1}$ the non increasing rearrangement of the absolute values of its components, that is $\amp_i(x)=|x_{p(i)}|$ for a permutation $p$ such that $\amp_1(x)\geq \amp_2(x) \geq \ldots \geq \amp_K(x)\geq 0$. Then the mapping $\amp$ together with the probability measure $\nu$ on $S^{\natoms-1}$ induces a pull-back probability measure $\nu_\amp$ on $\amp(S^{\natoms-1})$, by $\nu_\amp (\Omega):=\nu(\amp^{-1}(\Omega))$ for any measurable set $\Omega\subseteq \amp(S^{\natoms-1})$. With the help of this new measure we can rewrite the expectations we need to calculate as
\begin{align}
\E_y \left( \max_{|I| = S} \| \dico_I^\star y\|_1 \right)& 
= \E_x \left( \max_{|I| = S} \| \dico_I^\star \dico x\|_1 \right)\notag\\
& = \int_{x} \max_{|I| = S} \| \dico_I^\star \dico x\|_1 d\nu
= \int_{\amp(x)} \E_p \E_\sigma \left( \max_{|I| = S} \|\dico^\star \dico \amp_{p,\sigma}(x)\|_1 \right)  d\nu_\amp.\notag
\end{align}
The expectation inside the integral should seem familiar. Indeed we have calculated it already in the proof of Proposition~\ref{th:exactsimple} for $\amp(x)$ a fixed decaying sequence satisfying $\amp_S(x) > \amp_{S+1}(x) + 2 \mu \|x\|_1$. Since this property is satisfied almost surely we have
\begin{align}
\E_y \left( \max_{|I| = S} \| \dico_I^\star y\|_1 \right) & 
= \int_{\amp(x)} \E_p \E_\sigma \left( \max_{|I| = S} \|\dico^\star \dico \amp_{p,\sigma}(x)\|_1 \right)  d\nu_\amp \notag\\
&=\int_{\amp(x)} c_1(x) +\ldots+ c_S(x)  d\nu_\amp:= \bar c_1 +\ldots +\bar c_S.\notag
\end{align}
For the expectation of a perturbed dictionary $\pdico$ we get in analogy
\begin{align}
\E_y \left( \max_{|I| = S} \| \pdico_I^\star y\|_1 \right) & 
= \int_{\amp(x)} \E_p \E_\sigma \left( \max_{|I| = S} \|\pdico^\star \dico \amp_{p,\sigma}(x)\|_1 \right)  d\nu_\amp\notag\\
&\leq \int_{\amp(x)} \eta(x) +\left( c_1(x)+ \ldots +c_S(x) \right)  \frac{1}{K} \sum_{i} \alpha_i \, d\nu_\amp,\notag
\end{align}
where
\begin{align}
\eta(x) := 4\eps S \sqrt{B}\sum_{i:\eps_i\neq 0} \exp \left(-\frac{(c_S(x)-c_{S+1}(x)- 2\mu\|x\|_1 - \frac{\eps^2}{2})^2}{8 \eps_{i}^2}\right).\notag
\end{align}
Since $ \amp_S(x)-\amp_{S+1}(x) - 2\mu \|x\|_1 \geq \beta $ almost surely we have
\begin{align}
\eta(x) \leq 4\eps S \sqrt{B}\sum_{i:\eps_i\neq 0} \exp \left(-\frac{(\beta- \frac{\eps^2}{2})^2}{8 \eps_{i}^2}\right):=\eta_\beta,\notag
\end{align}
almost surely and therefore
\begin{align}
\E_y \left( \max_{|I| = S} \| \pdico_I^\star y\|_1 \right) & \leq  \eta_\beta +\left( \bar c_1+ \ldots +\bar c_S \right)  \frac{1}{K} \sum_{i} \alpha_i.\notag
\end{align}
Following the same argument as in the proof of Proposition~\ref{th:exactsimple} we see that $\E_y \left( \max_{|I| = S} \| \dico_I^\star y\|_1 \right)>\E_y \left( \max_{|I| = S} \| \pdico_I^\star y\|_1 \right)$ as soon as 
\begin{align}
 \eps < \frac{\beta}{1 + 3\sqrt{ \log\left( \frac{25 K^2 S \sqrt{B}}{\beta (\bar c_1+ \ldots +\bar c_S)}\right) } }.\notag
\end{align}
\end{proof}

\subsection{Proof of Theorem~\ref{th:stableasympt}} \label{app:stableasympt}
Again we first reformulate and prove the theorem for the case of a symmetric coefficient distribution based on one sequence and then extend it with an integration argument.
\begin{proposition}\label{th:stablesimple}
Let $\dico$ be a unit norm frame with frame constants $A\leq B$ and coherence $\mu$.
Let $x \in \R^K$ be a random permutation of a sequence $\amp$, where $\amp_1 \geq \amp_2 \geq \amp_3 \ldots \geq \amp_\natoms \geq 0$ and $\|\amp\|_2=1$, provided with random $\pm$ signs, that is $x=\amp_{p, \sigma}$ with probability $\P(p,\sigma)=(2^\natoms \natoms!)^{-1}$. Further let $\noise=(\noise(1) \ldots \noise(d))$ be a centred random subgaussian noise-vector with parameter $\nsigma$ and assume that the signals are generated according to the noisy signal model in~\eqref{noisymodel}. If we have
\begin{align}\label{munoisecond}
\max\{\mu,\nsigma\} \leq \frac{\amp_S - \amp_{S+1}}{ \sqrt{ 72(\log a + \log\log a)} } \quad \mbox{for} \quad a= \frac{112 K^2 S (\sqrt{B}+1)}{C_\noise(c_S-c_{S+1})(c_1+ \ldots +c_S)},
\end{align}
where $C_\noise= \E_\noise \left((1+\|\noise\|_2^2)^{-1/2}\right)$, then there is a local maximum of \eqref{expcrit} at $\tilde{\pdico}$ satisfying,
\begin{align}
d(\tilde{\pdico},\dico) \leq \frac{12 SK^2 \sqrt{B}}{C_\noise(c_1+ \ldots +c_S)}  \exp\left( \frac{-(c_S - c_{S+1})^2}{72\max\{\mu^2,\nsigma^2\}} \right).\notag
\end{align} 
\end{proposition}
\begin{proof}
To prove the proposition we digress from the conventional scheme of first calculating the expectation of our objective function for both the original and a perturbed dictionary and then comparing and instead bound the difference of the expectations directly.
\begin{align}
\E_y &\left( \max_{|I| = S} \| \dico_I^\star y\|_1 \right)- \E_y \left( \max_{|I| = S} \| \pdico_I^\star y\|_1 \right)\notag\\
&= \E_{p,\sigma,\noise}  \left( \max_{|I| = S} \left\|  \frac{\dico_I^\star (\dico \amp_{p,\sigma}+\noise)}{\sqrt{1+\|\noise\|_2^2} }\right\|_1 - \max_{|I| = S} \left\|  \frac{\pdico_I^\star (\dico \amp_{p,\sigma}+\noise)}{\sqrt{1+\|\noise\|_2^2} }\right\|_1  \right)\notag\\
&=\E_{p,\sigma,\noise} \left( \frac{\max_{|I| = S} \left\| \dico_I^\star (\dico \amp_{p,\sigma}+\noise) \right\|_1 - \max_{|I| = S} \left\| \pdico_I^\star (\dico \amp_{p,\sigma}+\noise) \right\|_1}{\sqrt{1+\|\noise\|_2^2}}  \right):=\E_{p,\sigma,\noise}(\Delta_{p,\sigma,\noise})\notag
\end{align}
Again our strategy is to show that for a fixed $p$ for most $\sigma$ and $\noise$ the maximal response of both the original dictionary and the perturbation is attained at $I_p$. The expressions we therefore need to lower (upper) bound for $i\in I_p$ ($i \notin I_p$) are
\begin{align*} 
 |\ip{ \atom_{i}}{\dico \amp_{p,\sigma} +\noise }| &= \big |\sigma_i c_{p(i)} + \sum_{j\neq i} \sigma_j c_{p(j)}  \ip{\atom_i}{\atom_j} + \ip{\atom_i}{\noise}\big|,\\
  |\ip{ \patom_{i}}{\dico \amp_{p,\sigma} +\noise }| &=\big |\alpha_i \sigma_i c_{p(i)} + \alpha_i \sum_{j\neq i} \sigma_j c_{p(j)} \ip{\atom_i}{\atom_j} +\omega_i \ip{z_i}{\dico \amp_{p,\sigma}} +\ip{\patom_i}{\noise} \big|.
\end{align*}
However, instead of using a worst case estimate for the gap between the responses of the original dictionary inside and outside $I_p$, we now make use of the fact that for most sign sequences we have a gap size of order $c_S - c_{S+1}$.
This means that as soon as $|\sum_{j\neq i} \sigma_j c_{p(j)}  \ip{\atom_i}{\atom_j}|$, $\omega_i |\ip{z_i}{\dico \amp_{p,\sigma}}|$ and the noise related terms $|\ip{\atom_i}{\noise}|$ and $|\ip{\patom_i}{\noise}|$ are of order $(c_S - c_{S+1})$
the maximal response of both the original dictionary and the perturbation is attained at $I_p$.
In particular, defining the sets
\begin{align}
\Sigma_p:=\bigcup_{i} &\left\{  \sigma \mbox{ s.t. }  \Big| \sum_{j\neq i} \sigma_j \amp_{p(j)}\ip{\atom_{i}}{ \atom_j }\Big| \geq \frac{c_S - c_{S+1}}{6} \: \mbox{ or } \: \omega_i|\ip{z_i}{\dico \amp_{p,\sigma}}| \geq 
 \frac{ c_S - c_{S+1}-\frac{ 3\eps^2}{2} }{6}\right\},\notag
\end{align} for a fixed permutation $p$ and 
\begin{align}
R:=\bigcup_{i} &\left\{  r \mbox{ s.t. } | \ip{\atom_i}{\noise}|\geq \frac{c_S-c_{S+1}}{3} \: \mbox{ or } \: | \ip{\patom_i}{\noise}|\geq \frac{c_S-c_{S+1}}{6} \right\},\notag
\end{align}
we see that both maxima are attained at $I_p$ as long as $\sigma \notin \Sigma_p$ and $r \notin R$.
Using Hoeffding's inequality we get that
\begin{align}
\P\left(  \Big| \sum_{j\neq i} \sigma_j \amp_{p(j)}\ip{\atom_{i}}{ \atom_j }\Big| > t  \right) \notag
&\leq 2 \exp \left( \frac{-t^2}{2 \sum_{j\neq i} \amp^2_{p(j)} |\ip{\atom_{i}}{ \atom_j } |^2} \right) \leq 2 \exp\left( \frac{-t^2}{2\mu^2} \right),\notag
\end{align}
while from the proof of Proposition~\ref{th:exactsimple} we know that
for $\eps_{i}\neq 0$ we have
$
\P(\omega_{i} | \langle z_{i}, \dico \amp_{p,\sigma} \rangle|  \geq s ) \leq 2 \exp \left(-\frac{s^2}{2 \eps_{i}^2}\right)
$. Setting $t=(c_S - c_{S+1})/6$, $s=( c_S - c_{S+1}-\frac{3 \eps^2}{2})/6$ and using a union bound then leads to 
\begin{align}\label{Pbadsigma}
\P(\Sigma_p) \leq 2K\exp \left(-\frac{\big(c_S - c_{S+1}-\frac{3 \eps^2}{2} \big)^2 }{72 \eps^2 }\right) + 2K \exp\left( \frac{-(c_S - c_{S+1})^2}{72\mu^2} \right).
\end{align}
Since the $\noise(i)$ are subgaussian with parameter $\nsigma$ we have for any $v=(v_1\ldots v_d)$ and $t\geq0$, $\P(|\ip{v}{r}|\geq t)\leq \exp\left(-\frac{t^2}{2\nsigma^2 \|v\|_2^2}\right)$, see e.g. \cite{ve10}. Taking a union bound over all $\atom_i, \patom_i$ with the corresponding choice for $t$ then leads to the estimate
\begin{align}\label{Pbadnoise}
\P(R)\leq 2K \exp \left(-\frac{\big(c_S - c_{S+1}\big)^2 }{72 \nsigma^2 }\right) + 2K \exp \left(-\frac{\big(c_S - c_{S+1}\big)^2 }{18 \nsigma^2 }\right).
\end{align}
We now split the expectations over the sign and noise patterns for a fixed p to get
\begin{align}
 \E_\sigma \E_\noise(\Delta_{p,\sigma,\noise})
&= \E_\sigma \left( \int_{r\notin R}  \Delta_{p,\sigma,\noise} d\nu_r\right) +  \E_\sigma \left( \int_{r\in R}  \Delta_{p,\sigma,\noise} d\nu_r\right) \notag \\
&=\sum_{\sigma\notin \Sigma_p} \P(\sigma)  \int_{r\notin R}  \Delta_{p,\sigma,\noise} d\nu_r + \sum_{\sigma\in \Sigma_p} \P(\sigma)  \int_{r\notin R}  \Delta_{p,\sigma,\noise} d\nu_r\notag\\
&\hspace{7cm}+\E_\sigma \left( \int_{r\in R}  \Delta_{p,\sigma,\noise} d\nu_r\right). \label{eq:splitprob}
\end{align}
Next note that $\Sigma_p$ is symmetric in the sense that we either have $(\sigma_1, \ldots ,\pm \sigma_i,\ldots ,\sigma_K) \in \Sigma_p$ or $(\sigma_1, \ldots ,\pm \sigma_i,\ldots ,\sigma_K) \notin \Sigma_p$. 
Thus we get for the first term in \eqref{eq:splitprob},
 \begin{align} 
\sum_{\sigma\notin \Sigma_p} \P(\sigma)  \int_{r\notin R}  \Delta_{p,\sigma,\noise} d\nu_r &=  \int_{r\notin R} \sum_{\sigma\notin \Sigma_p} \P(\sigma)   \left( \frac{ \left\| \dico_{I_p}^\star (\dico \amp_{p,\sigma}+\noise) \right\|_1 -  \left\| \pdico_{I_p}^\star (\dico \amp_{p,\sigma}+\noise) \right\|_1}{\sqrt{1+\|\noise\|_2^2}}  \right) d\nu_r \notag\\
&=  \int_{r\notin R} \sum_{\sigma\notin \Sigma_p} \P(\sigma)   \left( \frac{ \sum_{i\in I_p} c_{p(i)} \frac{\eps_i^2}{2} }{\sqrt{1+\|\noise\|_2^2}}  \right) d\nu_r.\notag
\end{align}
To bound the last two terms in \eqref{eq:splitprob} we first find an upper bound for$\max_{|I| = S} \| \pdico_I^\star (\dico c_{p,\sigma} +r )\|_1$:
\begin{align}
\max_{|I| = S} \| \pdico_I^\star( \dico c_{p,\sigma} +r )\|_1
&= \max_{|I| = S} \sum_{i\in I} |\ip{\alpha_i \atom_i + \omega_i z_i} {\dico c_{p,\sigma} +r )}|\notag\\
&\leq \max_{|I| = S} \sum_{i\in I} \big(1-\frac{\eps_i^2}{2}\big) |\ip{\atom_i} {\dico c_{p,\sigma} +r )}| + \eps_i \| \dico c_{p,\sigma} +r \|_2\notag\\
&\leq \max_{|I| = S} \sum_{i\in I}  \big(1-\frac{\eps^2}{2}\big)|\ip{\atom_i} {\dico c_{p,\sigma} +r )}| + \eps\big(\sqrt{B}+\| r\|_2\big)\notag\\
&= \big(1-\frac{\eps^2}{2}\big)\max_{|I| = S} \| \dico_I^\star (\dico c_{p,\sigma} +r )\|_1+ \eps S \big(\sqrt{B}+\| r\|_2\big).\notag
\end{align}
This then leads to the following lower bound for $\Delta_{p,\sigma,\noise} $:
\begin{align}
\Delta_{p,\sigma,\noise}&\geq (1+\|\noise\|_2^2)^{-1/2}\left( \max_{|I| = S} \| \dico_I^\star (\dico c_{p,\sigma} +r )\|_1 \frac{\eps^2}{2}-\eps S \big(\sqrt{B}+\| r\|_2\big)\right)\notag\\
&\geq (1+\|\noise\|_2^2)^{-1/2}\left( \| \dico_{I_p}^\star (\dico c_{p,\sigma} +r )\|_1 \frac{\eps^2}{2}-\eps S \big(\sqrt{B}+\| r\|_2\big)\right).\notag
 \end{align}
Using again the symmetry of $\Sigma_p$ we have 
\begin{align}
\sum_{\sigma\in \Sigma_p} \P(\sigma)  \int_{r\notin R}  \Delta_{p,\sigma,\noise} d\nu_r &\geq \int_{r\notin R}  \sum_{\sigma\in \Sigma_p} \P(\sigma) \frac{ \| \dico_{I_p}^\star (\dico c_{p,\sigma} +r )\|_1 \frac{\eps^2}{2}-\eps S \big(\sqrt{B}+\| r\|_2\big)}{\sqrt{1+\|\noise\|_2^2}} d\nu_r\notag\\
&\geq \int_{r\notin R}  \sum_{\sigma\in \Sigma_p} \P(\sigma)\left( \frac{  \sum_{i\in I_p} c_{p(i)}  \frac{\eps^2}{2}}{\sqrt{1+\|\noise\|_2^2}} - \eps S \big(\sqrt{B}+1\big)\right)d\nu_r\notag\\
&\geq \int_{r\notin R}  \sum_{\sigma\in \Sigma_p} \P(\sigma)\left( \frac{  \sum_{i\in I_p} c_{p(i)}  \frac{\eps_i^2}{2}}{\sqrt{1+\|\noise\|_2^2}} - \eps S \big(\sqrt{B}+1\big)\right)d\nu_r,\notag
\end{align}
and similarly
\begin{align}
\E_\sigma \left( \int_{r\in R}  \Delta_{p,\sigma,\noise} d\nu_r\right) &\geq \int_{r \in R} \E_\sigma\left( \frac{ \| \dico_{I_p}^\star (\dico c_{p,\sigma} +r )\|_1 \frac{\eps^2}{2}-\eps S \big(\sqrt{B}+\| r\|_2\big)}{\sqrt{1+\|\noise\|_2^2}} \right)d\nu_r\notag\\
&\geq \int_{r \in R}\frac{  \sum_{i\in I_p} c_{p(i)}  \frac{\eps_i^2}{2}}{\sqrt{1+\|\noise\|_2^2}} - \eps S \big(\sqrt{B}+1\big) d\nu_r.\notag
\end{align}
Resubstituting into~\eqref{eq:splitprob} we get
\begin{align}
 \E_\sigma \E_\noise(\Delta_{p,\sigma,\noise})&\geq \int_{r \in R}\frac{  \sum_{i\in I_p} c_{p(i)}  \frac{\eps_i^2}{2}}{\sqrt{1+\|\noise\|_2^2}} - \eps S \big(\sqrt{B}+1\big) d\nu_r\notag\\
&\hspace{2cm}+ \int_{r\notin R}  \sum_{\sigma\in \Sigma_p} \P(\sigma)\left( \frac{  \sum_{i\in I_p} c_{p(i)}  \frac{\eps_i^2}{2}}{\sqrt{1+\|\noise\|_2^2}} - \eps S \big(\sqrt{B}+1\big)\right)d\nu_r\notag\\
&\hspace{4cm}+\int_{r\notin R} \sum_{\sigma\notin \Sigma_p} \P(\sigma)   \left( \frac{ \sum_{i\in I_p} c_{p(i)} \frac{\eps_i^2}{2} }{\sqrt{1+\|\noise\|_2^2}}  \right) d\nu_r \notag\\
&\geq \int_{r}\frac{  \sum_{i\in I_p} c_{p(i)}  \frac{\eps_i^2}{2}}{\sqrt{1+\|\noise\|_2^2}}  d\nu_r  - \eps S \big(\sqrt{B}+1\big)\cdot \big(P(R)+P(\Sigma_p)\big). \label{deltalowerbound}
\end{align}
Taking the expectation over the permutations then yields
\begin{align}
 \E_{p,\sigma,\noise} (\Delta_{p,\sigma,\noise})&\geq \E_\noise \E_p \left( \frac{  \sum_{i\in I_p} c_{p(i)}  \frac{\eps_i^2}{2}}{\sqrt{1+\|\noise\|_2^2}}\right)  - \eps S \big(\sqrt{B}+1\big)\cdot \big(\P(R)+\E_p \P(\Sigma_p)\big)\notag\\
 &\geq \E_\noise \left(\frac{1}{\sqrt{1+\|\noise\|_2^2}}\right) \frac{c_1+\ldots + c_S}{2K} \sum_i \eps_i^2  - \eps S \big(\sqrt{B}+1\big)\cdot \big(\P(R)+\E_p \P(\Sigma_p)\big).\notag
\end{align}
Using the probability estimates from \eqref{Pbadsigma}/\eqref{Pbadnoise} we see that $ \E_{p,\sigma,\noise} (\Delta_{p,\sigma,\noise})>0$ is implied by
\begin{align}
\eps \geq \frac{4SK^2\big(\sqrt{B}+1\big)}{C_\noise \gamma}\left(  \exp \left(\frac{-\big(\beta-\frac{3 \eps^2}{2} \big)^2 }{72 \eps^2 }\right) + \exp\left( \frac{-\beta^2}{72\mu^2} \right) + \exp\left(\frac{-\beta^2 }{72 \nsigma^2 }\right) + \exp \left(\frac{-\beta^2 }{18 \nsigma^2 }\right) \right),\notag
\end{align}
where we have used the abbreviations $\gamma=c_1+\ldots + c_S$, $\beta=c_S - c_{S+1}$ and $C_\noise= \E_\noise \left((1+\|\noise\|_2^2)^{-1/2}\right)$. We now
 proceed by splitting the above condition.
We define $\epsmin$ by asking that
\begin{align}
\frac{\eps}{3} \geq \frac{4SK^2\big(\sqrt{B}+1\big)}{C_\noise \gamma} \exp\left(- \frac{\beta^2}{72\max\{\mu^2,\nsigma^2\}} \right) :=\frac{\epsmin}{3}\notag
\end{align}
and $\epsmax$ implicitly by asking that
\begin{align}
\frac{\eps}{3} - \frac{\eps^4}{81} \geq \frac{4SK^2\big(\sqrt{B}+1\big)}{C_\noise \gamma} \exp \left(-\frac{\big(\beta-\frac{3 \eps^2}{2} \big)^2 }{72 \eps^2 }\right).\notag
\end{align}
Following the line of argument in the proof of Proposition~\ref{th:exactsimple} we see that the above condition is guaranteed as soon as 
\begin{align}
 \eps \leq \frac{\beta}{\frac{5}{2} + 9\sqrt{ \log\left( \frac{112 K^2 S( \sqrt{B}+1)}{C_\noise \beta \gamma}\right) } }:=\epsmax.\notag
 \end{align}
The statement follows from making sure that $\epsmin < \epsmax$.
\end{proof}

\begin{proof}[of Theorem~\ref{th:stableasympt}]
Using the pull-back probability measure $\nu_\amp$ we can write
\begin{align}
\E_y &\left( \max_{|I| = S} \| \dico_I^\star y\|_1 \right)- \E_y \left( \max_{|I| = S} \| \pdico_I^\star y\|_1 \right)
=\int_{c(x)} \E_{p,\sigma,\noise}\left (\Delta_{p,\sigma,\noise, c(x)}\right)d\nu_c,\notag
\end{align}
where $\Delta_{p,\sigma,\noise, c(x)}$ is defined analogue to $\Delta_{p,\sigma,\noise}$ in the last proof, that is replacing $c$ by $c(x)$.
The statement follows from employing the lower estimate for $\E_{p,\sigma,\noise}\left (\Delta_{p,\sigma,\noise, c(x)}\right)$ from \eqref{deltalowerbound} and replacing $c_1+\ldots +c_S$ by $\bar c_1+\ldots +\bar c_S$ resp. $c_S-c_{S+1}$ by its lower bound $\beta$ in the proof of Proposition~\ref{th:stablesimple}.
\end{proof}

\subsection{Proof of Theorems \ref{th:exactfinite} and \ref{th:stablefinite}} \label{app:finite}
Since the proofs of Theorems~\ref{th:exactfinite} and \ref{th:stablefinite} are conceptually equivalent we will 
combine them into one and just split the argument for the inevitable juggling of constants.
\begin{proof}
As outlined in the proof idea we need a Lipschitz property for the mapping $\pdico \rightarrow \frac{1}{N} \sum_{n=1}^N \max_{|I| = S} \| \pdico_I^\star y_n\|_1$ for both signal models, the concentration of the sum around its expectation for a $\delta$ net covering the space of all admissible dictionaries close to $\dico$ and a triangle inequality argument to get to the final statement.\\
To show the Lipschitz property we use a reverse triangle inequality:
\begin{align}
\left| \max_{|I| = S} \| \pdico_I^\star y_n\|_1-\max_{|I| = S} \| \ppdico_I^\star y_n\|_1\right| 
&= \left| \max_{|I| = S} \| \ppdico_I^\star y_n- (\ppdico_I^\star - \pdico_I^\star) y_n\|_1-\max_{|I| = S} \| \ppdico_I^\star y_n\|_1\right| \notag \\
&\leq \max_{|I| = S} \| (\ppdico_I^\star - \pdico_I^\star) y_n\|_1 \notag\\
&\leq S \max_k \|\patom_k-\ppatom_k\|_2 \|y_n\|_2 \notag \\
&\leq d(\pdico, \ppdico) S \big(\sqrt{B} +1\big).\notag
\end{align}
Note that for the noise-free signal model we can replace $ \big(\sqrt{B} +1\big)$ by $\sqrt{B}$ in the last expression. By averaging over $n$ we get that the mapping in question is
Lipschitz with constant $S \big(\sqrt{B} +1\big)$ in the noisy and $S\sqrt{B}$ in the noise-free case, that is
\begin{align}
\left|  \frac{1}{N} \sum_{n=1}^N \max_{|I| = S} \| \pdico_I^\star y_n\|_1- \frac{1}{N} \sum_{n=1}^N \max_{|I| = S} \| \ppdico_I^\star y_n\|_1\right| &\leq d(\pdico, \ppdico) S \big(\sqrt{B} +1\big).\notag
\end{align}
To show that the averaged sums concentrate around their expectations we use our favourite tool Hoeffding's inequality. Set $X_n=\max_{|I| = S} \| \dico_I^\star y_n\|_1 - \max_{|I| = S} \| \pdico_I^\star y_n\|_1$, then we have $|X_n| \leq  \eps S \big(\sqrt{B} +1\big)$, resp. $|X_n| \leq \eps S \sqrt{B}$ in the noise-free case, and get the estimate
\begin{align}
&\P\left( \left| \frac{1}{N} \sum_{n=1}^N\left( \max_{|I| = S} \| \dico_I^\star y_n\|_1- \max_{|I| = S} \| \pdico_I^\star y_n\|_1\right) - \E \left(\max_{|I| = S} \| \dico_I^\star y_1\|_1 - \max_{|I| = S} \| \pdico_I^\star y_1\|_1  \right) \right| \geq 2 t \right) \notag \\ 
&\hspace{8cm} \leq 2\exp\left( \frac{-2Nt^2}{\eps^2 S^2 \big(\sqrt{B} +1\big)^2} \right).\notag
\end{align}
Next we need to choose a $\delta$-net for all perturbations $\pdico$ with $d(\dico,\pdico)\leq \epsmax$, that is a finite set of perturbations $\net$ such that for every $\pdico$ we can find
$\ppdico \in \net$ with $d(\pdico,\ppdico)\leq \delta$. Recalling the parametrisation of all $\eps$-perturbations from the proof of Proposition~\ref{th:exactsimple}, we see that the space we need to cover is included in the product of K balls with radius $\epsmax$ in dimension $\ddim$. From e.g. Lemma~2 in \cite{ve10} we know that for the $\ddim$ dimensional ball of radius $\epsmax$ we can find a $\delta$-net $\net_d$ satisfying $\sharp \net_d \leq \left(\epsmax + \frac{2\epsmax}{\delta}\right)^\ddim$, so for our space of $\eps$-perturbations we can find a $\delta$-net $\net$ satisfying
\begin{align}
\sharp \net \leq \left(\epsmax + \frac{2\epsmax}{\delta}\right)^{\natoms \ddim} \leq  \left( \frac{3\epsmax}{\delta}\right)^{K \ddim}.\notag
\end{align}
Taking a union bound we can now estimate the probability that we have concentration for all perturbations in the net as
\begin{align}
\P&\left( \exists \pdico \in \net: \left| \frac{1}{N} \sum_{n=1}^N\left( \max_{|I| = S} \| \dico_I^\star y_n\|_1- \max_{|I| = S} \| \pdico_I^\star y_n\|_1\right) \right. \right.\notag \\
&\hspace{6cm} \left. \left.\phantom{\frac{1}{N} \sum_{n=1}^N}- \E \left(\max_{|I| = S} \| \dico_I^\star y_1\|_1 - \max_{|I| = S} \| \pdico_I^\star y_1\|_1  \right) \right|  \geq 2t \right)\notag \\
&\hspace{5cm}\leq  \left( \frac{3\epsmax}{\delta}\right)^{K \ddim} 2\exp\left( \frac{-2Nt^2}{\epsmax^2 S^2 \big(\sqrt{B} +1\big)^2} \right).\notag
\end{align}
Finally we are ready for the triangle inequality argument. For any $\pdico$ with $d(\pdico,\dico)=\eps \leq \epsmax$ we can find $\ppdico \in \net$ with $d(\ppdico,\pdico) \leq \delta$ and $d(\dico,\ppdico)=\bar\eps$ and therefore get
\begin{align}
&\frac{1}{N} \sum_{n=1}^N \max_{|I| = S} \| \dico_I^\star y_n\|_1 - \frac{1}{N} \sum_{n=1}^N \max_{|I| = S} \| \pdico_I^\star y_n\|_1\notag\\
&= \frac{1}{N} \sum_{n=1}^N \max_{|I| = S} \| \dico_I^\star y_n\|_1 - \E \left(\max_{|I| = S} \| \dico_I^\star y_1\|_1  \right)
+ \E \left(\max_{|I| = S} \| \dico_I^\star y_1\|_1\right)  - \E \left(\max_{|I| = S} \| \ppdico_I^\star y_1\|_1\right) \notag\\
&+ \E \left(\max_{|I| = S} \| \ppdico_I^\star y_1\|_1\right) -  \frac{1}{N} \sum_{n=1}^N \max_{|I| = S} \| \ppdico_I^\star y_n\|_1
 +  \frac{1}{N} \sum_{n=1}^N \max_{|I| = S} \| \ppdico_I^\star y_n\|_1 -  \frac{1}{N} \sum_{n=1}^N \max_{|I| = S} \| \pdico_I^\star y_n\|_1\notag\\
&\hspace{3cm} \geq  \E \left(\max_{|I| = S} \| \dico_I^\star y_1\|_1\right)  - \E \left(\max_{|I| = S} \| \ppdico_I^\star y_1\|_1\right) - 2t - \delta S  \big(\sqrt{B} +1\big).\notag
\end{align}
Depending on the signal model we now have to substitute the values for the asymptotic differences $\E \left(\max_{|I| = S} \| \dico_I^\star y_1\|_1\right)  - \E \left(\max_{|I| = S} \| \ppdico_I^\star y_1\|_1\right)$   
calculated in the previous proofs. \\
Under the conditions given in Theorem~\ref{th:exactfinite} we have,
\begin{align}
\frac{1}{N} \sum_{n=1}^N &\max_{|I| = S} \| \dico_I^\star y_n\|_1 - \frac{1}{N} \sum_{n=1}^N \max_{|I| = S} \| \pdico_I^\star y_n\|_1\notag \\
&\geq
\bar \eps^2 \frac{\bar c_1+ \ldots +\bar c_S}{2K}-  4\bar \eps S K \sqrt{B}  \exp \left(-\frac{(\beta - \frac{\bar \eps^2}{2})^2}{8\bar\eps^2 }\right) - 2t - \delta S  \sqrt{B}. \label{finitediffexact}
\end{align}
To make sure that the above expression is larger than zero, we split it into two conditions. The first
condition
\begin{align}
\bar \eps^2 \frac{\bar c_1+ \ldots +\bar c_S}{4K} > 4\bar \eps S K \sqrt{B}  \exp \left(-\frac{(\beta - \frac{\bar \eps^2}{2})^2}{8\bar\eps^2 }\right)\notag
\end{align}
is satisfied as soon as 
\begin{align}
\bar \eps \leq \frac{\beta}{\frac{25\sqrt{8}}{99}\left(1 + 4\sqrt{ \log\left( \frac{50 K^2 S \sqrt{B}}{\beta (\bar c_1+ \ldots +\bar c_S)}\right) }\right)}.\notag
\end{align}
To concretise the second condition
\begin{align}
\bar \eps^2 \frac{\bar c_1+ \ldots +\bar c_S}{4K}\geq 2t + \delta S  \sqrt{B},\notag
\end{align}
we choose $t=\tilde\eps^2 \frac{\bar c_1+ \ldots +\bar c_S}{16K}$ and $\delta=\tilde \eps^2 \frac{\bar c_1+ \ldots +\bar c_S}{8KS \sqrt{B}}$ to arrive at $\bar\eps\geq 2 \tilde \eps$. 
Given that $\bar \eps $ differs at most by $\delta$ from $\eps$ we see that \eqref{finitediffexact} is larger than zero except with probability
\begin{align}
2 \exp\left(- \frac{N \tilde \eps^4 (\bar c_1+ \ldots +\bar c_S)^2}{128\epsmax^2 S^2K^2 B} + Kd \log  \left( \frac{24\epsmax KS \sqrt{B}  }{\tilde \eps^2 (\bar c_1+ \ldots +\bar c_S) }\right)\right), \notag
\end{align}
as long as 
\begin{align}\label{minmaxineq}
\epsmin:=  \tilde\eps +\frac{ \tilde\eps^2}{8K} \leq \eps \leq \epsmax \leq \frac{\beta}{\frac{25\sqrt{8}}{99}\left(1 + 4\sqrt{ \log\left( \frac{50 K^2 S \sqrt{B}}{\beta (\bar c_1+ \ldots +\bar c_S)}\right) }\right)} - \frac{ \tilde \eps^2}{8K}.
\end{align}
While for the asymptotic results we tried to make $\epsmax$ as large as possible to indicate how large the basin of attraction could be, for the finite sample size results we want it as small as possible in order to keep the sampling complexity small and therefore choose $\epsmax=  \epsmin$. The statement then follows from making sure that the right most inequality in~\eqref{minmaxineq} is satisfied and simplifications.\\
In case of the noisy signal model, that is under the conditions given in Theorem~\ref{th:stablefinite}, we have
\begin{align}
\frac{1}{N}& \sum_{n=1}^N \max_{|I| = S} \| \dico_I^\star y_n\|_1 - \frac{1}{N} \sum_{n=1}^N \max_{|I| = S} \| \pdico_I^\star y_n\|_1\notag \\
&\geq
\bar \eps^2 \frac{\bar c_1+ \ldots +\bar c_S}{2 C_\noise K} - 2t - \delta S \big(\sqrt{B}+1\big)- 2\bar \eps S K \big(\sqrt{B}+1\big)\cdot \notag \\
& \qquad \cdot \left(  \exp \left(\frac{-\big(\beta-\frac{ 3\bar\eps^2}{2} \big)^2 }{72 \bar \eps^2 }\right) + \exp\left( \frac{-\beta^2}{72\mu^2} \right) + \exp\left(\frac{-\beta^2 }{72 \nsigma^2 }\right) + \exp \left(-\frac{\beta^2 }{18 \nsigma^2 }\right) \right).
\label{finitediffnoisy}
\end{align}
Splitting equally gives us four conditions:
\begin{align}
\bar \eps &\geq \frac{16 C_\noise S K^2\big(\sqrt{B}+1\big)}{\bar c_1+ \ldots +\bar c_S} \exp\left(- \frac{\beta^2}{72\mu^2} \right):=\eps_\mu,\notag\\
\bar \eps &\geq \frac{16 C_\noise S K^2\big(\sqrt{B}+1\big)}{\bar c_1+ \ldots +\bar c_S} \exp\left(- \frac{\beta^2}{72\nsigma^2} \right):=\eps_\nsigma, \notag\\
\bar \eps^2 &\geq  \frac{8 C_\noise K}{\bar c_1+ \ldots +\bar c_S} \left(2t + \delta S \big(\sqrt{B}+1\big)\right), \notag \\
 \bar \eps (1- \frac{\bar \eps^3}{64})&> \frac{16 C_\noise S K^2\big(\sqrt{B}+1\big)}{\bar c_1+ \ldots +\bar c_S} \exp\left(-\frac{\big(\beta-\frac{ 3\bar\eps^2}{2} \big)^2 }{72 \bar \eps^2 }\right)\label{finitesampleseps}.
\end{align}
Choosing $t=\tilde\eps^2_{\mu,\nsigma} \frac{\bar c_1+ \ldots +\bar c_S}{32 C_\noise K}$ and $\delta=\tilde\eps^2_{\mu,\nsigma} \frac{\bar c_1+ \ldots +\bar c_S}{16KS (\sqrt{B}+1)}$ we can  merge the first three conditions to $\bar \eps\geq \tilde\eps^2_{\mu,\nsigma}$, while following the usual argument, Condition \eqref{finitesampleseps} is satisfied once
\begin{align}
\bar \eps \leq \frac{\beta}{\frac{9}{4} + 9\sqrt{ \log\left( \frac{150 K^2 S \big(\sqrt{B}+1\big)}{\beta C_\noise (\bar c_1+ \ldots +\bar c_S)}\right) } }.\notag
\end{align}
Given that $\bar \eps $ differs at most by $\delta$ from $\eps$ we see that \eqref{finitediffnoisy} is larger than zero except with probability
\begin{align}
 2\exp\left(- \frac{N\tilde\eps^4_{\mu,\nsigma}(\bar c_1+ \ldots +\bar c_S)^2}{512 \epsmax^2 C_\noise^2 S^2 K^2  \big(\sqrt{B}+1\big)^2} + Kd \log  \left( \frac{48 \epsmax KS \big(\sqrt{B}+1\big) }{\tilde\eps^2_{\mu,\nsigma} (\bar c_1+ \ldots +\bar c_S)} \right)\right),\notag
\end{align}
as long as 
\begin{align}\label{minmaxineq2}
\epsmin:=\tilde\eps_{\mu,\nsigma} +\frac{\tilde\eps^2_{\mu,\nsigma}}{16 K} \leq  \eps \leq \epsmax \leq \frac{\beta}{\frac{9}{4} + 9\sqrt{ \log\left( \frac{150 K^2 S \big(\sqrt{B}+1\big)}{\beta C_\noise (\bar c_1+ \ldots +\bar c_S)}\right) } } - \frac{ N^{-2q}}{K}.
\end{align}
Again the statement follows from choosing $\epsmax =  \epsmin$, making sure that the right most inequality in~\eqref{minmaxineq2} is satisfied and simplifications. 
\end{proof}

\bibliography{liod}
\end{document}